\documentclass[12pt]{article}

\def\ismaindoc{}
\def\arxivbuild{}
\usepackage{preamble}

\title{Credible Auctions via MPC Gadgets: Bounding Information Leakage Under Abort}
\author{Matheus V. X. Ferreira\\
University of Virginia, USA}
\date{}
\usepackage{review-mode}

\begin{document}

\begin{titlepage}

	\maketitle

	\begin{abstract}
		
The design of credible auctions---mechanisms where a revenue-maximizing auctioneer has no incentive to deviate from the protocol---faces a fundamental cryptographic barrier when the auctioneer controls shill bidders.
While a natural approach is to use Secure Multi-Party Computation (MPC) to remove the trusted auctioneer, the impossibility of fair coin flipping of \citet{Cleve86} implies that monolithic MPC protocols grant the auctioneer a ``free option'': they can learn the auction's outcome and unilaterally abort if the revenue is unsatisfactory.
Cryptographic commitments with ex-ante penalties mitigate this abort asymmetry, but no finite penalty suffices for heavy-tailed distributions.
We circumvent this barrier by introducing the \emph{MPC Decomposition Principle}.
Rather than encrypting the entire mechanism, we use MPC strictly as an information-restriction tool.
We isolate the \emph{winner determination problem} into a minimal MPC gadget that computes and reveals the winner's identity but no payment information.
This qualitative restriction mathematically bounds the information leaked upon an abort.
By combining this gadget with sequential revelation and finite economic penalties, we design the \emph{Sequential Revelation Auction} (SRA).
We prove that bounding the information leakage strictly bounds the value of the free option, showing that a penalty of $\pen \geq \sum_{i=1}^n \rev(F_i)$ is sufficient for credibility, and tight: for equal-revenue distributions, every smaller penalty admits a profitable deviation.
Using constant-round MPC, the SRA resolves an open question of \citet{AkbarpourLi} and \citet{FerreiraWeinberg} by providing a constant-round, incentive-compatible, revenue-optimal credible auction for all product distributions with vanishing revenue tails.

	\end{abstract}


	\newpage
	\tableofcontents

\end{titlepage}

\section{Introduction}

The study of mechanism design traditionally assumes a trusted intermediary who faithfully implements the mechanism's rules.
In practice---ranging from large-scale digital advertising to decentralized finance~\citep{Bogetoft09}---the auctioneer is a rational agent who may deviate if doing so is both undetectable and increases their revenue.
The literature on \emph{credible auctions}, initiated by \citet{AkbarpourLi}, formalizes this problem.
In the full information setting (where the auctioneer can observe all the actions from bidders), they proved a fundamental impossibility: the ascending price auction is the unique revenue-optimal, incentive-compatible auction without undetectable profitable deviations, a desideratum we call credibility.
Unfortunately, this auction requires an unbounded amount of communication in the worst case.

To circumvent this impossibility, \citet{FerreiraWeinberg} propose transforming the auction into a game of incomplete information using cryptography.
Their \emph{Deferred Revelation Auction (DRA)} requires bidders to cryptographically commit to their valuations, revealing them only after the auctioneer commits to the execution path.
Moreover, participants deposit collaterals to incentivize opening of commitments.
However, this approach hits a fundamental barrier when applied to bidder valuations drawn from heavy-tailed distributions (e.g., the equal-revenue distribution), where \citet{FerreiraWeinberg} prove that no finite ex-ante penalty can deter a malicious auctioneer.
In the DRA, this failure occurs because the auctioneer's shill bidder can wait for honest bidders to reveal their commitments, observe their valuations, and then refuse to open its own commitment.
This forced abort triggers a restart where the auctioneer now knows the bids.
For heavy-tailed distributions, the expected revenue gained from exploiting this acquired knowledge is infinite, meaning no finite collateral can deter the abort.

A natural alternative is to abandon commit-reveal schemes and use Secure Multi-Party Computation (MPC) to encrypt the entire auction.
However, this monolithic MPC approach hits a fundamental cryptographic barrier first identified by \citet{Cleve86} in the context of fair coin flipping: the \emph{abort problem}.
They show that in protocols without an honest majority (which includes an auctioneer commanding unbounded shill bids), an adversary can always learn their output and unilaterally abort the protocol before honest parties receive theirs.

In economic terms, this abort asymmetry grants the auctioneer a ``free option''.
If the auctioneer encrypts the entire optimal mechanism (allocation and payments) within an MPC protocol, they can run the protocol, observe the provisional revenue, and if it is unsatisfactory, abort the execution and restart.
Since the abort reveals that all active bidders have valuations below a certain threshold, the restart is not an independent trial but a biased one.
This allows the auctioneer to safely explore the bidders' valuations using shills, effectively extracting full surplus and violating the credibility of the mechanism.

In this work, we introduce the \emph{Sequential Revelation Auction (SRA)}, which resolves this impossibility and achieves credibility for \emph{any} product distribution over bidder valuations.
We overcome the free option problem by introducing the \emph{MPC Decomposition Principle}.
Rather than following the traditional paradigm of ``encrypting everything,'' we use MPC strictly as an \emph{information-restriction tool}.
Instead of computing the entire mechanism inside the MPC, we isolate the smallest possible piece of information necessary: the identity of the winner.
We introduce a \emph{Winner Determination Gadget (WDG)}, a minimal MPC functionality that takes committed bids and solves the Winner Determination Problem (WDP) by outputting \emph{only} the identity of the winner, revealing \emph{no payments} or \emph{bids}.

This decomposition shifts the information asymmetry.
If the auctioneer uses a shill to probe the price and subsequently aborts the gadget, they learn only \emph{who} was winning, but not \emph{how much} bidders were willing to pay.
We formally quantify the value of this leaked information.
Our main technical contribution, the Information Leakage Lemma (\Cref{lem:information_leakage}), bounds the extra revenue the auctioneer can extract from the knowledge acquired during an abort.
Therefore, by setting an ex-ante abort penalty $\pen$ to at least this threshold, we ensure that the cost of aborting always exceeds the value of the free option.

\subsection{Technical Overview}

We assume the existence of an MPC protocol with identifiable abort that implements any ideal functionality $f$.
This protocol allows a set of parties to compute any ideal functionality $f(v_1, \ldots, v_n) = (y_1, \ldots, y_n)$ where $v_i$ is the input that player $i$ contributes to evaluating $f$ and $y_i$ is the output only they observe.
In this model, any player $i$ can abort the execution at any time learning $y_i$ and preventing others from learning $y_{-i}$; however, if they do so, all other parties learn the identity of the deviating player.

In the SRA, we instantiate $f$ as the Winner Determination Problem: given a set of bids $\bm v$, output the highest bidder $i^*$.
Then, we ask the winner to reveal their bid.
The only undetectable deviation available to the auctioneer is to modify the collection of bids $\bm v$ by adding a set of shill bids $r_1 < r_2 < \dots < r_k$ into the input and, if a shill is chosen as the winner, to refuse to reveal its bid, in which case the protocol recomputes the WDP without it.
Our proof reduces the auctioneer's adaptive strategies to pure threshold strategies, which abort every shill above some $r_i$ and reveal $r_i$ and below, and then decomposes the revenue of each by the shill at which the descent stops, charging the information each abort leaks against the penalty it pays.
By bounding the information leaked during such an abort, we show that setting the penalty $\pen \ge \sum_{i=1}^n \rev(F_i)$ ensures that aborting is never profitable where $(F_1, \ldots, F_n)$ is the value distribution profile of real bidders and $\rev(F_i)$ is the optimal revenue (\Cref{lem:myerson}) an auctioneer can obtain in a single-bidder auction with a bidder with valuation drawn from $F_i$.
\subsection{Our Contribution}

We provide a constant-round, incentive-compatible, revenue-optimal auction that is credible for all product distributions over bidder valuations with vanishing revenue tails, resolving the open problem left by \citet{FerreiraWeinberg}.
Our contributions are threefold:
\begin{enumerate}
	\item \textbf{The MPC Decomposition Principle:}
	      We introduce a new design methodology for credible auctions based on \emph{decomposing} the auction computation.
	      Rather than implementing the entire mechanism (allocation and payments) inside MPC---which, as we show in \Cref{sec:strawman}, inevitably leaks enough information to violate credibility---we isolate the minimal subroutine needed: a Winner Determination Gadget that computes \emph{only} the Winner Determination Problem.
	      Payments are then determined \emph{outside} the MPC through sequential revelation.
	      This decomposition principle is the conceptual contribution: it shows that for credibility, minimizing the computed function is more important than encrypting everything.
	\item \textbf{General Credibility:}
	      We prove that for any product distribution $\bm{F}$ with vanishing revenue tails, the SRA is revenue-optimal, incentive-compatible, and credible provided the penalty $\pen$ satisfies $\pen \ge \sum_{i=1}^n \rev(F_i)$ (\Cref{thm:main}).
	      This overcomes the fundamental barrier identified by \citet{FerreiraWeinberg}, who exhibit a single non-MHR bidder---the equal-revenue distribution---for which \emph{no} finite penalty makes the Deferred Revelation Auction credible.
	      Our protocol circumvents this impossibility by using MPC to change the information structure of the game.
	\item \textbf{Tightness and Round Complexity:}
	      We show that the penalty bound $\pen \ge \sum_{i=1}^n \rev(F_i)$ is tight: for equal-revenue distributions truncated at a large enough $H$, any penalty $\pen < \sum_{i=1}^n \rev(F_i)$ admits profitable deviations (\Cref{thm:lower_bound}).
	      Additionally, the SRA runs in constant rounds on the equilibrium path (assuming standard constant-round MPC constructs).
	      This complements the ascending protocol of \citet{FerreiraEssaidi}, which achieves general credibility for $n$ \iid bidders in constant rounds \emph{in expectation} but has unbounded worst-case round complexity.
	      Our guarantee differs along two axes: the round bound on the equilibrium path is deterministic rather than in expectation, and our credibility guarantee holds for arbitrary \emph{independent, non-identical} distributions $F_1, \ldots, F_n$ rather than the \iid case.
\end{enumerate}

\paragraph*{Practical Collateral Requirements.}
The threshold $\pen \ge \sum_{i=1}^n \rev(F_i)$ is practically reasonable, not merely finite.
On the canonical stress test---the equal-revenue distribution $F(v) = 1 - 1/v$ for $v \ge 1$, where \citet{FerreiraWeinberg} show that no finite collateral makes the Deferred Revelation Auction credible, even against a single bidder---the SRA is credible at $\pen = 1$ per bidder once the distribution is truncated as in \Cref{sec:lower}, and since $v \ge 1$ on the support, no bidder ever posts collateral larger than their own valuation.
Write $\ivv_i$ for the ironed virtual value of bidder $i$, the transformation of \citet{Myerson81} under which the optimal auction allocates to the highest transformed bid (\Cref{sec:prelim}); the positive part of $\ivv_i(v_i)$ has expectation $\rev(F_i)$.
Since the item goes to at most one bidder, $\OPT(\bm F) = \e{\max_i \ivv_i(v_i)^+} \le \sum_{i=1}^n \e{\ivv_i(v_i)^+} = \sum_{i=1}^n \rev(F_i)$, so the requirement is $\OPT(\bm F) \le \pen \le n \cdot \OPT(\bm F)$: between one and $n$ times the revenue the auction generates.
It is also a \emph{liquidity} requirement rather than an expected cost, since every bidder who follows the protocol recovers $\pen$ in full.

\subsection{Related Work}

\paragraph*{Secure Multi-Party Computation.}
Secure Multi-Party Computation (MPC) was introduced by \citet{Yao82} with the ``Millionaires' Problem'' and \citet{Blum81} with ``Coin Flipping by Telephone.
''
The first large-scale commercial application of MPC was the Danish Sugar Beet Auction \citep{Bogetoft09}, where farmers traded production contracts.
In this system, robustness against bidder dropout was handled by re-running the computation or assuming the servers were honest.
This approach to \emph{safety} was acceptable because the threat model did not involve a malicious auctioneer.
Our work, on the other hand, focuses on \emph{credibility}, where the auctioneer is the primary adversary.
We show that the restart mechanism, essential for safety, becomes a vulnerability for credibility.
For our auction to have small round complexity, we rely on works showing that any functionality can be realized in a constant number of rounds.
Two rounds suffice for security \emph{with abort}~\citep{GargGentryHaleviRaykova14,MukherjeeWichs16}, relying on powerful cryptographic primitives like \emph{Multi-Key Fully-Homomorphic Encryption} (MK-FHE) and \emph{indistinguishability obfuscation} (iO).
Our protocol requires the stronger guarantee of \emph{identifiable} abort, since the penalty must be charged to a named party, and identifiability provably costs additional rounds: \citet{CiampiRaviSiniscalchiWaldner22} show that four rounds are necessary and sufficient in the plain model.
Constant-round protocols with identifiable abort are known~\citep{IOZ14,BaumOrsini20,Cohen23}, which is all our result requires.
While these primitives are currently computationally expensive, they establish the theoretical feasibility of avoiding the high round complexity of previous credible mechanisms~\citep{FerreiraEssaidi}.

\paragraph*{Credible Auctions.}
The credibility framework was established by \citet{AkbarpourLi}, who show that no revenue-optimal auction is simultaneously strategy-proof, credible, and static.
\citet{FerreiraWeinberg} introduce cryptographic assumptions to the problem, constructing two-round revenue optimal credible auctions for MHR distributions.
However, their negative result for heavy-tailed distributions left open the question of general feasibility.
Our work resolves this open problem.
\citet{FerreiraEssaidi} also achieve general credibility, for $n$ \iid bidders with a finite monopoly price, in constant rounds in expectation but with unbounded worst-case round complexity.
Our SRA, on the other hand, runs in constant rounds (assuming standard constant-round MPC constructs) on the equilibrium path, and handles independent but non-identical distributions.
Like theirs, our worst-case round complexity is unbounded: the number of rounds grows linearly with the number of bidders that abort, but aborting is a dominated strategy for all bidders (real or shills).

\citet{ChitraFerreira} extend the DRA to settings where bidders can broadcast messages, which is relevant to our MPC-based approach since MPC with identifiable abort requires broadcast channels.
\citet{GaneshZhang} extend the DRA to multi-item settings with matroid constraints but also require similar distributional assumptions to \citet{FerreiraWeinberg} and \citet{FerreiraEssaidi}.

\paragraph*{Deviations by the Party Running the Mechanism.}
Credibility is one instance of a broader question: how to measure the incentive of the party who runs a mechanism to deviate from its announced rules when that party also profits from the outcome.
\citet{FerreiraParkes23} ask it of a decentralized exchange, where the block proposer both sequences the trades and trades on its own account, and give a verifiable sequencing rule under which, for every user, either the trade executes at a price at least as good as if it were the only trade in the block, or the proposer has no risk-free profit from reordering it.
Transaction fee mechanism design asks the same question of the miner who runs a blockchain's fee auction: \citet{Roughgarden24} formalizes incentive compatibility for a myopic miner who may add fake transactions, alongside incentive compatibility for users and resistance to off-chain agreements between the two; \citet{ChungShi23} prove that no non-trivial mechanism is simultaneously incentive compatible for users and resilient to a coalition of the miner with a user, and \citet{ChungRoughgardenShi24} study which relaxations of that collusion resilience remain achievable.
\citet{ShiChungWu23} then show that running the mechanism inside an MPC among the miners relaxes the impossibility.
In all of these, a deviation is judged by the profit it earns the party running the mechanism while also considering collusion of buyers with miners, as in \Cref{def:credibility}; credibility additionally restricts attention to the deviations the bidders cannot detect.

There is a rich literature on implementing auction mechanisms using MPC \citep{NaorPinkasSumner99, Brandt06}.
Most of this work focuses on \emph{privacy} (keeping bids secret) or \emph{fairness} (preventing early bid peeking).
Our focus on \emph{credibility} is distinct: we view the auctioneer as the adversary and specifically target the strategic implications of the abort capability.
Our result highlights that for credibility, ``encrypting everything'' is not the solution; rather, credibility requires minimizing the function the MPC computes.

\subsection{Paper Organization}
\Cref{sec:prelim} defines the credibility framework.
\Cref{sec:strawman} illustrates the failure of black-box MPC.
\Cref{sec:protocol} details the SRA protocol.
\Cref{sec:example} provides a warmup by showing that, for the equal-revenue distribution, a natural class of shill strategies earns the auctioneer no more than the honest revenue.
\Cref{sec:info_leakage} bounds the information an abort leaks, and \Cref{sec:credible} uses that bound to prove credibility for every product distribution, which addresses the open problem of \citet{FerreiraWeinberg}.
\Cref{sec:lower} provides the matching lower bound.
\submission{Proofs omitted for space, together with a worked equal-revenue warmup, appear in the full version.}

\begin{standalonebib}
	\bibliographystyle{ACM-Reference-Format}
	\bibliography{mybib}
\end{standalonebib}

\section{Preliminaries}\label{sec:prelim}

We consider a single-item auction setting with a set of honest, risk-neutral bidders $N = \{1, \ldots, n\}$.
Additionally, there exists a countable set of potential shill bidders $S = \{n+1, n+2, \ldots\}$.
Honest bidders $i \in N$ have private valuations $v_i \in V_i \subseteq \mathbb{R}_{\geq 0}$ drawn independently from cumulative distribution functions $F_i$ with probability density functions $f_i$.
Shill bidders $s \in S$ have a fixed valuation $v_s = 0$.
Crucially, shill bidders share the auctioneer's utility function (revenue maximization) and do not behave as self-interested agents maximizing their own allocation utility.
We denote the valuation profile of honest agents by $\bm{v} = (v_1, \ldots, v_n) \in \bm{V} = \times_{i \in N} V_i$.
For value profile $\bm{v} = (v_i)_{i \in N \cup S}$ and subset $A \subseteq N \cup S$, we write $\bm{v}_A = (v_i)_{i \in A}$ for the restriction of $\bm{v}$ to bidders in $A$.
For any function $f$, we write $f^+ \coloneqq \max\{f, 0\}$ for the positive part.

\begin{definition}[Mechanism]
	A direct revelation mechanism $M = (x, p)$ consists of an \emph{allocation rule} $x: \bm{V} \to [0, 1]^n$, where $x_i(v)$ is the probability that bidder $i$ receives the item and feasibility requires $\sum_{i \in N} x_i(v) \le 1$ for all $v \in \bm{V}$, together with a \emph{payment rule} $p: \bm{V} \to \mathbb{R}^n$, where $p_i(v)$ is the expected payment from bidder $i$.
\end{definition}

We assume agents have quasi-linear utility functions.

\begin{definition}[Utility and DSIC]
	The utility of bidder $i$ with valuation $v_i$ who reports $v'_i$ while the others report $v_{-i}$ is $u_i(v'_i, v_{-i}; v_i) = v_i \cdot x_i(v'_i, v_{-i}) - p_i(v'_i, v_{-i})$.
	A mechanism is \emph{Dominant Strategy Incentive Compatible (DSIC)} if truthful reporting maximizes utility for every honest player $i \in N$, valuation $v_i$, and profile $v_{-i}$: $u_i(v_i, v_{-i}; v_i) \ge u_i(v'_i, v_{-i}; v_i)$ for all $v'_i \in V_i$.
\end{definition}

\begin{theorem}[\citet{Myerson81}]
	\label{thm:myerson_char}
	A direct revelation mechanism $(x, p)$ is DSIC if and only if:
	\begin{enumerate}
		\item \textbf{Allocation rule is monotone.}
		      For every bidder $i$ and declarations $v_{-i}$, $x_i(v_i, v_{-i})$ is non-decreasing in $v_i$.
		\item \textbf{Payment identity.}
		      The payment rule is determined by the allocation rule: for every bidder $i$, $v_i$, and $v_{-i}$, $p_i(v_i, v_{-i}) = v_i \cdot x_i(v_i, v_{-i}) - \int_0^{v_i} x_i(z, v_{-i}) dz$.
	\end{enumerate}
\end{theorem}

\begin{definition}[Revenue]
	\label{def:revenue}
	The \emph{ex-post revenue} of a mechanism $M = (x, p)$ for valuation profile $\bm{v}$ is the sum of payments, $R^M(\bm{v}) = \sum_{i \in N} p_i(\bm{v})$; its \emph{expected revenue} is $\rev^M(\bm F) = \e[\bm{v} \sim \bm F]{R^M(\bm{v})}$, and we write $R(\cdot)$ for short when $M$ is clear from context.
	We write $\OPT(\bm F) = \max_{M \text{ DSIC}} \rev^M(\bm F)$ for the optimal expected revenue over DSIC mechanisms, and, for a \emph{single} distribution $F_i$, $\rev(F_i) = \sup_{p \ge 0} p \cdot (1 - F_i(p))$ for the optimal \emph{monopoly} revenue obtainable from bidder $i$ alone by posting a price.
	Throughout, $\rev(\cdot)$ applied to a single distribution always denotes this monopoly revenue, while revenue of a mechanism always carries the mechanism in the superscript.
\end{definition}

Throughout we assume each $F_i$ has a vanishing revenue tail, $\lim_{v \to \infty} v\,(1 - F_i(v)) = 0$; this is the hypothesis of \Cref{lem:myerson} below, and it implies $\rev(F_i) < \infty$, so the penalty $\sum_i \rev(F_i)$ is finite.

\begin{definition}[Regular Distribution]
	A differentiable cumulative distribution function $F$ with density $f$ is said to be \emph{regular} if the virtual valuation function $\phi(v) = v - \frac{1-F(v)}{f(v)}$ is monotone non-decreasing.
\end{definition}

\begin{definition}[Ironed Virtual Valuation]
	For any distribution $F$, the \emph{ironed virtual valuation} function $\bar{\phi}(v)$ is constructed by applying Myerson's ironing procedure to $\phi(v)$ to ensure monotonicity.
	Specifically, let $H(q) = \int_0^q \phi(F^{-1}(t)) dt$ for $q \in [0, 1]$.
	We define $G$ as the convex hull of $H$, i.e., the largest convex function such that $G(q) \le H(q)$ for all $q$.
	Then, $\bar{\phi}(v) = G'(F(v))$.
	If $F$ is regular, $\bar{\phi}(v) = \phi(v)$.
	We define the \emph{inverse} $\bar{\phi}^{-1}(y) := \inf\{v : \bar{\phi}(v) \geq y\}$, which is well-defined even when $\bar{\phi}$ has flat regions due to ironing.
	An \emph{ironed interval} of $F$ is a maximal interval $[a, b]$ of values on whose interior $G < H$; $\bar{\phi}$ is constant on it and differs from $\phi$ somewhere in it, and outside the ironed intervals $\bar{\phi} = \phi$.
	The \emph{ironed bid} $\ironedbid(v)$ is $v$ when $v$ lies in no ironed interval and the left endpoint $a$ of the ironed interval $[a, b]$ containing $v$ otherwise; it is non-decreasing and constant on every ironed interval.
\end{definition}

\begin{lemma}[Virtual surplus identity; {\citet{Myerson81}}]
	\label{lem:myerson}
	Let $M = (x, p)$ be a DSIC mechanism allocating a single item to agents with independent private values drawn from distributions $F_i$ satisfying $\lim_{v \to \infty} v\,(1 - F_i(v)) = 0$.
	Then the expected revenue equals the expected virtual surplus and is at most the expected ironed virtual surplus: \[ \rev^M(\bm F) = \e[\bm{v} \sim \bm F]{ \sum_{i} \phi_i(v_i) x_i(\bm{v}) } \le \e[\bm{v} \sim \bm F]{ \sum_{i} \ivv_i(v_i) x_i(\bm{v}) }, \] where $\phi_i$ and $\ivv_i$ are the virtual valuation and the ironed virtual valuation functions of bidder $i$.
	The inequality is an equality whenever, for every bidder $i$ and every $\bm v_{-i}$, the allocation $x_i(\cdot, \bm v_{-i})$ is constant on every ironed interval of $F_i$.
\end{lemma}

Myerson's result identifies the revenue-optimal auction as the mechanism that allocates the item to the bidder with the highest ironed virtual value, provided it is non-negative.
Ironing makes ties a positive-measure event rather than a null one---$\ivv_i$ is flat on each ironed interval, so every profile landing inside one produces a tie---and the Winner Determination Gadget of \Cref{sec:protocol} publishes which of the tied parties is provisionally winning.
We therefore fix the comparison and the tie-break together, as a single ranking:

\begin{definition}[Ranking]
	\label{def:ranking}
	For \emph{participants} $i \ne j$ submitting bids $b_i, b_j$, we write $\ivv_i(b_i) < \ivv_j(b_j)$ to mean that the mechanism ranks $j$ above $i$: either $\ivv_i(b_i)$ is numerically smaller than $\ivv_j(b_j)$; or the two are numerically equal and $\ironedbid_i(b_i) < \ironedbid_j(b_j)$; or both of these are equalities and $j < i$.
	We write $\le$ for the reflexive version.
	A higher ironed virtual value therefore wins, ties go to the higher ironed bid, and any remaining tie goes to the smaller index; the last clause separates any two distinct participants, so this is a strict total order and every finite set of participants has a unique maximum.

	Two conventions keep this unambiguous.
	First, the overload applies only between \emph{two participants}' virtual values; a comparison against a constant, such as $\ivv_i(v_i) \ge 0$, is always the ordinary numeric one.
	Second, the relation is determined by the pairs $(i, b_i)$ and $(j, b_j)$ rather than by the two numbers alone, so it is antisymmetric even where the numbers coincide: $\ivv_i(b_i) \le \ivv_j(b_j)$ and $\ivv_j(b_j) \le \ivv_i(b_i)$ hold together only when $i = j$.
	The bids and identities are always clear from context.

	The ranking is \emph{monotone in a participant's own bid}: for $z < z'$ we have $\ivv_i(z) \le \ivv_i(z')$ and $\ironedbid_i(z) \le \ironedbid_i(z')$ numerically because both are non-decreasing, and either the first inequality is strict, in which case the first clause ranks $z'$ above $z$, or it is an equality and the second clause ranks $z'$ weakly above $z$, strictly unless $z$ and $z'$ lie in a common ironed interval.
	A participant's rank is therefore constant on each ironed interval of its own distribution.
\end{definition}

\begin{definition}[Myerson mechanism]
	\label{def:myerson_mech}
	The mechanism allocates to the highest-ranked participant whose ironed virtual value is non-negative, $x_i(\bm{b}) = \ind{\ivv_i(b_i) \ge 0} \cdot \ind{\ivv_j(b_j) < \ivv_i(b_i) \ \forall j \ne i}$, and charges the winner the threshold payment induced by that rule, $p_i(\bm{b}) = \inf\{ z \ge 0 : x_i(z, \bm{b}_{-i}) = 1 \}$, losers paying nothing.
\end{definition}

Winning under \Cref{def:myerson_mech} requires $i$ to outrank the fixed best competitor and to satisfy $\ivv_i(b_i) \ge 0$; both conditions are upward closed in $i$'s own bid by the monotonicity of \Cref{def:ranking}, so the allocation rule is monotone and, with its threshold payments, the mechanism is DSIC by \Cref{thm:myerson_char}.
Both conditions depend on $b_i$ only through the rank of $i$, which is constant on each ironed interval of $F_i$, so \Cref{lem:myerson} holds with equality and the expected revenue is the expected ironed virtual surplus $\e{\max_i \ivv_i(v_i)^+} = \OPT(\bm F)$.
Thus the the Myerson mechanism is revenue-optimal.

The protocol design problem sets \emph{safety} (the protocol should handle dropouts gracefully) against \emph{credibility} (the auctioneer should not be able to exploit the protocol); the definitions below make this tension precise.

\subsection{Credibility}
In this section, we formalize the notion of credibility we consider.
Intuitively, a mechanism is credible if the auctioneer cannot profit by deviating from the protocol in a way that is indistinguishable from honest behavior.

\begin{definition}[Auction Protocol]
	An \emph{auction protocol} is a pair $(M, \Pi)$, where $M = (x, p)$ is a mechanism and $\Pi$ is a cryptographic protocol that purportedly realizes $M$.
\end{definition}

\begin{definition}[Safe Deviation~\cite{AkbarpourLi}]
	A deviation by the auctioneer in protocol $\Pi$ is \emph{safe} if, for every honest bidder $i$, there exists a strategy profile of the other bidders $N \cup S \setminus \{i\}$ such that $i$'s view under the deviation is computationally indistinguishable from $i$'s view under the honest execution of $\Pi$ with that strategy profile.
\end{definition}

To illustrate, consider a second-price auction.
The deviation where the auctioneer shill bids is a \emph{safe} deviation since from any real bidder's perspective, the shill is indistinguishable from a real bidder.
On the other hand, an \emph{unsafe} deviation is one where the auctioneer refuses to allocate the item to the winner.
This aims to capture an auctioneer that is willing to deviate as long as deviations cannot be detected by honest bidders.

\begin{definition}[Credibility]\label{def:credibility}
	The \emph{revenue of a deviation} is the expected sum of the payments the auctioneer receives from the real bidders minus the penalties the protocol charges the auctioneer's shills, if any; a transfer from a shill to the auctioneer nets to zero.
	An auction protocol $(M, \Pi)$ is \emph{credible} if there exists no safe deviation for the auctioneer within $\Pi$ whose revenue is strictly higher than the expected revenue of the honest execution of $M$ without shills.
\end{definition}

\noindent We can now state our primary design goal:

\begin{definition}[Optimal Credible Auction Design Problem]
	The goal is to find an auction protocol $(M, \Pi)$ maximizing $\rev^M(\bm F)$ subject to $M$ being DSIC for the honest bidders $N$ and $(M, \Pi)$ being credible.
\end{definition}
One further requirement is what makes that problem hard, and it is the reason the strawman of \Cref{sec:strawman} fails.
A protocol must remain usable when a participant simply stops responding.

\begin{definition}[Safety against Honest Abort]
	\label{def:safety}
	A protocol $\Pi$ satisfies \emph{Safety against Honest Abort} if for any set of bidders $P$, if a non-auctioneer player $i \in P$ aborts, the protocol can continue and the expected revenue---that is, the expected sum of payments from the remaining bidders---is at least the expected revenue of the mechanism executed on the set $P \setminus \{i\}$.
\end{definition}

To satisfy \Cref{def:safety}, a protocol must ensure that a single bidder's refusal to participate does not destroy the auction's value, which in practice requires it to restart or to switch to a mode that ignores the aborting party's input.
Paradoxically, this creates the following vulnerability: it hands the auctioneer a ``free option'' to abort (perhaps through a shill) whenever the outcome is unfavorable, knowing the protocol \emph{must} offer a recourse that preserves revenue---so the auctioneer retries the auction having learned the outcome of the first attempt.
\Cref{sec:strawman} turns this observation into a concrete example.

\subsection{Cryptography}

This section introduces basic cryptographic definitions used in the paper.

\begin{definition}[Cryptographic Commitment]
	A commitment scheme lets a committer publish $c = \commit{v}{r}$, a \emph{commitment} to a message $v$ under randomness $r$.
	We require it to be \emph{computationally hiding}---for distinct $v, v'$ the distributions $\{\commit{v}{r}\}_r$ and $\{\commit{v'}{r}\}_r$ are computationally indistinguishable---\emph{perfectly binding}---no two pairs $(v, r) \ne (v', r')$ with $v \ne v'$ satisfy $\commit{v}{r} = \commit{v'}{r'}$---and \emph{non-malleable}, so that from $c$ no adversary can produce a commitment whose value bears a non-trivial relation to $v$.
	See \citet{Goldreich04} for the formal definitions.
\end{definition}

A \emph{multi-party computation (MPC) with identifiable abort} is a protocol that realizes an ideal functionality $\mathcal{F}$ that takes inputs $\bm{z} = (z_1, \ldots, z_n) \in (\mathcal{Z} \cup \{\bot\})^n$ from $n$ parties and computes a function $f(\bm{z}) = (y_1, \dots, y_n)$, where $y_i \in \mathcal{Y} \cup \{\bot\}$.
An input of $\bot$ means that the party contributed with no input, and an output of $\bot$ means that the party gets no output.

While an ideal MPC would guarantee output delivery to all parties, this is known to be impossible in the presence of a dishonest majority \citep{Cleve86}.
To build intuition, consider the simple ``Coin Flipping'' problem where two parties want to generate a uniform random bit $r = r_1 \oplus r_2$.
In the ideal world, a trusted party receives $r_1, r_2$ and sends $r$ to both.
In any real protocol, one party must send their value (or open their commitment) last.
If this last party sees the result and does not like it (e.g., they wanted Heads but the result is Tails), they can simply abort.
This allows the last party to bias the output (forcing an abort is equivalent to rejecting the coin flip).
This impossibility result implies that in any auction protocol with a potentially malicious auctioneer (dishonest majority), we must allow for the possibility of \emph{abort}.

This is particularly relevant in the context of auctions, where the auctioneer is allowed to control any number of shill bidders.
We therefore rely on the standard relaxation of security with \emph{identifiable abort}, which guarantees that if the adversary prevents honest parties from learning the output, their identity is revealed.

\begin{definition}[Ideal Functionality]
	Let $\F: \mathcal{Z}^n \to \mathcal{Y}^n$ be an $n$-party function.
	Ideally, a trusted party computes $(y_1, \dots, y_n) = \F(z_1, \dots, z_n)$ and sends $y_i$ to party $P_i$.
\end{definition}
We use the standard simulation-based notion of security with identifiable abort, in which the adversary may abort at any point during the execution; we state the ideal experiment explicitly, since the timing of the abort is what drives our incentive analysis.

\begin{definition}[Ideal model with identifiable abort]
	\label{def:ideal_ia}
	Fix an $n$-party ideal functionality $\F$ and a set $I \subset [n]$ of corrupted parties.
	In the \emph{ideal execution}, a simulator $\mathcal{S}$ controlling the parties in $I$ interacts with the trusted party computing $\F$ as follows.
	\begin{enumerate}
		\item Each honest party $i \notin I$ sends its input $z_i$ to the trusted party; $\mathcal{S}$ sends inputs $\{z_i\}_{i \in I}$ of its choice (possibly $\bot$).
		\item The trusted party computes $(y_1, \ldots, y_n) = f(\bm z)$ and sends $\{y_i\}_{i \in I}$ to $\mathcal{S}$.
		\item \label{step:sim_decision} $\mathcal{S}$ replies with either $\mathsf{continue}$, in which case each honest party $i$ receives $y_i$; or $(\mathsf{abort}, i^*)$ for some $i^* \in I$ of its choice, in which case every party outputs $(\bot, i^*)$.
	\end{enumerate}
\end{definition}

\begin{definition}[Realization with Identifiable Abort]
	\label{def:realization_ia}
	A protocol $\Pi$ \emph{securely realizes $\F$ with identifiable abort} if for every probabilistic polynomial-time adversary $\mathcal{A}$ statically corrupting a set $I \subset [n]$ in the real execution of $\Pi$, there exists a probabilistic polynomial-time simulator $\mathcal{S}$ in the ideal model of \Cref{def:ideal_ia} such that the joint distribution of the honest parties' outputs and the adversary's view is computationally indistinguishable in the two executions.
	See \citet{Goldreich04} and \citet{IOZ14} for the formal treatment.
\end{definition}

Note that \Cref{def:ideal_ia} already encodes the worst case for our analysis: the adversary learns its own outputs $\{y_i\}_{i \in I}$ \emph{before} deciding whether honest parties receive theirs, and an abort names a corrupted party.
Every deviation we analyze in \Cref{sec:credible} is a choice of $\mathsf{continue}$ versus $(\mathsf{abort}, i^*)$ at \Cref{step:sim_decision} of \Cref{def:ideal_ia}, so no generality is lost by restricting attention to that decision.

\begin{theorem}[\citet{IOZ14, BaumOrsini20, Cohen23}]
	\label{thm:mpc_rounds}
	Under standard cryptographic assumptions and the availability of a broadcast channel, for any ideal functionality $\F$, there exists a protocol $\Pi$ that securely realizes $\F$ with identifiable abort in the presence of a static adversary controlling any number of corrupted parties $t < n$ in a \emph{constant} number of rounds of communication.
\end{theorem}

\begin{standalonebib}
	\bibliographystyle{ACM-Reference-Format}
	\bibliography{mybib}
\end{standalonebib}

\section{Strawman}\label{sec:strawman}

In this section, we analyze the credibility of a direct attempt to realize a revenue-optimal auction using generic Multi-Party Computation (MPC).
We define the ideal functionality $\F_{OPT}$ for the single-item Myerson auction---which receives bids, computes the optimal allocation and payments, and distributes the results---and demonstrate that any protocol realizing $\F_{OPT}$ while satisfying \emph{Safety against Honest Abort} (\Cref{def:safety}) is vulnerable to a welfare-extraction attack by a malicious auctioneer.

\begin{algorithm}[H]\label{alg:f_opt}
	\caption{Optimal Auction Functionality $\F_{OPT}$}
	\SetAlgoNoLine
	\DontPrintSemicolon
	\KwIn{Bid profile $\bm b = (b_i)_{i \in P}$ for participants $P$, each registered with a distribution $F_i$.}
	\textbf{Computation:}\;
	\begin{enumerate}
		\item Determine the allocation $x$ of the Myerson mechanism (\Cref{def:myerson_mech}): $x_i(\bm b) = 1$ if $\ivv_i(b_i) \geq 0$ and $\ivv_j(b_j) < \ivv_i(b_i)$ for every $j \in P \setminus \{i\}$ in the ranking of \Cref{def:ranking}, and $x_i(\bm b) = 0$ otherwise.
		\item Compute the threshold payment $p_i(\bm b) = \inf \{ z \geq 0 : x_i(z, \bm b_{-i}) = 1 \}$ if $x_i(\bm b) = 1$, and $p_i(\bm b) = 0$ otherwise.
	\end{enumerate}
	\Return{$(\bm x, \bm p)$.}
\end{algorithm}

We now consider a protocol $\Pi$ that realizes $\F_{OPT}$.
To satisfy \emph{Safety against Honest Abort} (\Cref{def:safety})---ensuring that a single bidder's dropout allows the auctioneer to recover the revenue from the remaining bidders---practical implementations (like \citealt{Bogetoft09}) must be able to exclude the aborting party and complete the auction (e.g., via restart).

\begin{algorithm}[H]\label{alg:strawman}
	\caption{Protocol that realizes $\F_{OPT}$ with Safety Restart}
	\SetAlgoNoLine
	\DontPrintSemicolon
	\KwIn{Bid profile $\bm b = (b_1, \dots, b_n)$ for bidders $P$.}
	\KwIn{MPC protocol $\Pi$ that realizes $\F_{OPT}$ with identifiable abort}

	\textbf{Round 1:}\;
	Parties invoke $\Pi$ on inputs $\bm b$.\;
	\If{$\Pi$ outputs $\bot$ due to abort by $k \in P$}{
		Identify aborter $k$.\;
		Update $P \leftarrow P \setminus \{k\}$.\;
		\textbf{Restart} round 1 with remaining players $P$.
	}
	\Return $(\bm x, \bm p)$ of $\Pi$ to all participants.\;
\end{algorithm}

\begin{theorem}\label{thm:impossibility_restart}
	Let $F$ be a regular distribution with bounded support $[\underline{v}, \bar{v}]$ and monopoly price $r^* = \ivv^{-1}(0)$, and let $\Pi$ be a protocol that realizes $\F_{OPT}$ for a single bidder with valuation $v$ drawn from $F$.
	If $\Pi$ satisfies \emph{Safety against Honest Abort} (\Cref{def:safety}), then for every $\epsilon > 0$ there exists a safe deviation for the auctioneer whose payment from the bidder is at least $v - \epsilon$ whenever $v \geq r^*$.
\end{theorem}

The honest mechanism charges the bidder $r^*$ on the event $v \geq r^*$ and nothing otherwise, so the deviation extracts the bidder's entire surplus above the reserve, up to $\epsilon$.
The auctioneer uses one shill to probe the payment, which under the Myerson mechanism reveals the bidder's valuation, and then aborts to exploit that information in the restart.

\begin{proof}
	Fix $\epsilon > 0$.
	We define the deviation as follows.
	The auctioneer registers two shills $S_1$ and $S_2$ with the bidder's own distribution $F$; in the first execution of $\Pi$, $S_1$ bids $\underline{v}$ and $S_2$ bids $\bar{v}$.
	Note $S_2$ outranks the bidder: $S_2$ shares the bidder's transformation, and since $F$ is regular nothing is ironed, so the ranking of \Cref{def:ranking} settles the tie between equal ironed virtual values on the raw bid.
	Note $\ivv(\bar{v}) = \bar{v} \geq 0$, since $1 - F(\bar{v}) = 0$.
	Thus $S_2$ wins the first execution, and the auctioneer observes the payment charged to $S_2$, which is the least bid $z$ with $\ivv(z) \geq 0$ that outranks the bidder's bid $v$, that is, $\max\{v, r^*\}$.
	The auctioneer then instructs $S_2$ to abort.
	By safety, Algorithm~\ref{alg:strawman} restarts $\Pi$ with $S_2$ excluded and $S_1$ still active, and since it restarts from scratch, $S_1$ may submit a fresh bid.
	On the event $v \geq r^*$ the auctioneer has learned $v$ and sets the bid of $S_1$ to $v - \epsilon$; the bidder outranks $S_1$, wins, and pays the threshold $\max\{v - \epsilon, r^*\} \geq v - \epsilon$.
	The deviation is safe: a real bidder may also abort, so the bidder's view of the first execution is that of an honest execution in which another participant aborted, and its view of the second is that of an honest execution against a bidder with valuation $v - \epsilon$.
	This proves that the deviation extracts at least $v - \epsilon$ whenever $v \geq r^*$.
\end{proof}

This impossibility result reveals the fundamental weakness of monolithic MPC for auction design: the protocol reveals the \emph{entire outcome} (winner and price) to the adversary before honest parties can act.
In particular, the adversary learns the highest honest valuation through the payment identity, and the safety restart mechanism provides a costless option to exploit this information.

\paragraph*{Why committing bids does not help.}
One might attempt to fix this by requiring bids to be cryptographically committed, so that the auctioneer cannot modify shill bids after learning the outcome.
However, this defense fails: the auctioneer commits to $k$ shill bidders with bids $1, 2, \ldots, k$ and aborts all shills bidding above the highest honest bidder, probing the honest valuation at no cost.

\Cref{sec:protocol} presents the Sequential Revelation Auction, which avoids this vulnerability by decomposing the computation into minimal gadgets that reveal only the winner identity---never the price---and by imposing penalties that make aborting unprofitable.

\begin{standalonebib}
	\bibliographystyle{ACM-Reference-Format}
	\bibliography{mybib}
\end{standalonebib}

\section{The Protocol}\label{sec:protocol}

The strawman analysis (\Cref{sec:strawman}) shows that treating the MPC as a black box that realizes an auction functionality does not yield credibility: the adversary learns the full outcome upon abort, providing a ``free option'' to retry with extra information.
Our protocol takes a different approach.
Rather than compute the entire auction mechanism---allocation \emph{and} payments---inside MPC, we decompose the computation into a minimal \emph{gadget} that reveals only the winner identity, and we compute payments \emph{outside} the MPC via simultaneous revelation of committed bids.
This design ensures that an abort reveals no payment information, and economic penalties make aborting unprofitable.

The auctioneer's ``free option'' in the strawman protocol stems from learning payment information upon abort.
If the MPC computed only the winner identity---but not the payment---an abort would reveal who is winning, but not how much they would pay.
Since payments, not allocations, directly determine revenue, this limited information alone is not enough for the auctioneer to construct a profitable deviation, provided the restart is costly.
We formalize this idea via the \emph{Winner Determination Gadget}.

\begin{definition}[Winner determination gadget]
	\label{def:wdg}
	Let $\bm{c} = (c_i)_{i \in A}$ be a vector of commitments, broadcast by the bidders in the commit phase, \emph{before} the gadget is invoked (Step~\ref{step:commit} of Algorithm~\ref{alg:credible_protocol}).
	The \emph{Winner Determination Gadget} $\F_{\text{WDG}}^{\bm{c}}$ is the ideal functionality, parameterized by $\bm{c}$, by the distributions $(F_i)_{i \in A}$, and by the Myerson-optimal allocation rule $x^*$ (\Cref{def:myerson_mech}), that behaves as follows:
	\begin{enumerate}
		\item \textbf{Input:}
		      Each party $i \in A$ submits a private input $(v_i, \rho_i)$.
		\item \textbf{Computation:}
		      Let $A' = \{i \in A : c_i = \commit{v_i}{\rho_i}\}$ be the set of parties whose opening is consistent with their commitment.
		      Parties in $A \setminus A'$ are treated as having contributed no input.
		      The functionality computes the allocation $x^*(\bm{v}_{A'})$ and the identity $i^*$ of the winner, if any.
		\item \textbf{Output:}
		      Every party $i \in A$ receives the pair $\bigl(i^*, A \setminus A'\bigr)$, where $i^* = \bot$ if the item is unallocated.
	\end{enumerate}
\end{definition}

First, the winner's identity is a \emph{public} output, which the protocol needs---Algorithm~\ref{alg:credible_protocol} must detect when the provisional winner refuses to reveal, and the penalty can only be charged to a party the honest bidders can name---and which is harmless for credibility: wherever a real bidder wins, every shill is outbid and revealing weakly dominates aborting (\Cref{lem:reveal_outbid}), so the auctioneer never acts on the extra information and \Cref{lem:information_leakage} is unaffected.
Second, the functionality does not itself abort.
Commitment verification is a computation on the parties' inputs, not a protocol event: a party with an inconsistent opening is treated as absent and its identity reported.
Aborts are handled in the ideal model of \Cref{def:ideal_ia}, where the adversary may, after learning its own output, send $(\mathsf{abort}, i^*)$ for a corrupted $i^*$, in which case all parties output $(\bot, \{i^*\})$.
The gadget therefore reveals nothing about payments, or about others' valuations beyond what the allocation implies.
A party in $A \setminus A'$ is excluded by Algorithm~\ref{alg:credible_protocol} and forfeits its collateral like any other party that aborts, so an inconsistent opening is never a free way out of the auction.

\begin{algorithm}[htbp]
	\SetAlgoNoLine
	\LinesNumbered
	\DontPrintSemicolon
	\KwParameters{Penalty $\pen$ per abort, MPC protocol $\Pi_{\text{WDG}}$ that realizes $\F_{\text{WDG}}$, Myerson mechanism $(x, p)$.}
	\KwIn{Valuations $v_i$ from bidders $i \in A$.}
	\KwSetup{Each participant $i \in A$ registers a public distribution $F_i$, which fixes the ironed virtual value $\ivv_i$ and the ironed bid $\ironedbid_i$ by which the gadget ranks $i$ (\Cref{def:ranking}).}
	\KwOut{An ex-post allocation $x(v_A)$ and payment $p(v_A)$.}
	\textbf{Commit Phase:}\;
	Each bidder $i \in A$ deposits a collateral $\pen$\label{step:deposit} in escrow\;
	Each bidder $i \in A$ samples randomness $\rho_i$ and broadcasts $c_i = \commit{v_i}{\rho_i}$\label{step:commit}\;
	\textbf{Reveal Phase:}\;
	Invoke $\Pi_{\text{WDG}}$ with input $\{(v_i, \rho_i, c_i)\}_{i \in A}$\;
	\If{$\Pi_{\text{WDG}}$ aborts}{
		Let $k$ be the faulty bidder identified by $\Pi_{\text{WDG}}$\;
		Exclude $k$: $A \leftarrow A \setminus \{k\}$\;
		\textbf{Restart from Reveal Phase} with surviving bidders $A$\;
	}
	\Else{
		$\Pi_{\text{WDG}}$ publicly outputs the identity $i^*$ of the \emph{provisional winner} and the set $A \setminus A'$ of parties whose opening was inconsistent with their commitment\;
		Exclude the inconsistent parties: $A \leftarrow A'$\label{step:exclude-inconsistent}\;
		\If{$i^* = \bot$}{
			The item is unsold; the auctioneer refunds the collateral of every $i \in A$, burns the rest, and the protocol halts\;
		}
		The provisional winner reveals their commitment by broadcasting $(v_{i^*}, \rho_{i^*})$, then transfers $v_{i^*}$ to the auctioneer\label{step:reveal}\; \If{the provisional winner refuses to participate}{ Exclude the provisional winner: $A \leftarrow A \setminus \{i^*\}$\; \textbf{Restart from Reveal Phase} with remaining bidders $A$\; } \Else{ The provisional winner becomes the \emph{final winner}.
			All bidders $j \in A \setminus \{i^*\}$ reveal their commitment by broadcasting $(v_j, \rho_j)$\;
			\For{Each $j \in A \setminus \{i^*\}$}{
				\If{$j$ fails to reveal}{
					Exclude $j$: $A \leftarrow A \setminus \{j\}$\;
				}
			}
			Auctioneer refunds $v_{i^*} - p_{i^*}(v_A)$ to the winner (where $A$ is the set of currently active bidders)\label{step:refund}\;
			The auctioneer returns the collateral to each bidder that did not abort (i.e., $i \in A$), and \textbf{burns} the collateral for all others\label{step:burn}\;
		}
	}
	\caption{Sequential revelation auction (SRA)}
	\label{alg:credible_protocol}
\end{algorithm}

\begin{example}[Walk-through of the SRA]
	\label{ex:sra_walkthrough}
	Consider a single real bidder with valuation $v = 5$ drawn from the equal-revenue distribution $F(v) = 1 - 1/v$ on $[1, \infty)$, for which $\ivv(v) = v - (1 - F(v))/f(v) = v - v = 0$ for every $v$ in the support.
	The optimal reserve is $r^* = 1$ and $\rev(F) = 1$.
	The auctioneer introduces two shills with bids $r_1 = 2$ and $r_2 = 8$, committing them to the same transformation $\ivv$ as the real bidder.
	Since $\ivv \equiv 0$, all three bids carry the same ironed virtual value, and since $F$ is regular nothing is ironed and the ironed bid is the bid itself, so the ranking of \Cref{def:ranking} allocates the item to the highest raw bid.

	\textbf{Round 1.}
	The WDG is invoked with all three bids and identifies shill $r_2 = 8$ as the provisional winner.
	The auctioneer must decide between two options.
	\emph{Reveal} $r_2$: the shill wins, no real bidder is allocated the item, and the auctioneer collects revenue $0$---this is the cost of committing to a price floor of $8$.
	\emph{Hide} $r_2$: abort, forfeit the penalty $\pen$, and restart without $r_2$.

	Suppose the auctioneer hides $r_2$.
	Penalty $\pen$ is burned, and the auction restarts with $\{\text{bidder}, r_1\}$.

	\textbf{Round 2.}
	The WDG runs on $\{\text{bidder}, r_1\}$.
	Since $v = 5 > r_1 = 2$, the real bidder wins.
	The bidder reveals $v = 5$ and deposits $5$.
	Then $r_1$ is revealed as $2$.
	The payment is the threshold of \Cref{def:myerson_mech}---the least bid at which the bidder still wins---and since all bids tie at $\ivv \equiv 0$ that threshold is settled on the raw bid, giving $r_1 = 2$ rather than the $\ivv^{-1}(\ivv(r_1)) = 1$ a numeric inverse would suggest.
	The auctioneer refunds $5 - 2 = 3$, keeping revenue $2$.

	On this particular realization the honest mechanism collects the reserve $1$, whereas the deviation nets $2 - \pen$, which exceeds $1$ whenever $\pen < 1$ and ties it at $\pen = 1$---so \emph{ex post} the abort looks attractive.
	The decision to hide $r_2$, however, must be made \emph{before} $v$ is known, and at that point the auctioneer is conditioning only on the event that all real bids fall below $8$.
	In expectation over $v$, the continuation revenue after the abort is at most $\rev(F) = 1$, so the net expected payoff from hiding is at most $1 - \pen \leq 0$.
	Thus the auctioneer has no incentive to hide $r_2$ in the first place.
	\Cref{sec:info_leakage} is precisely the general form of this calculation: it bounds the value of the information the auctioneer obtains by aborting.
\end{example}

\subsection{Safety}

The requirement that the winner $i^*$ deposits $v_{i^*}$ \emph{before} other bidders reveal their commitments is structural, not incidental.
It resolves a potential conflict between safety and credibility.
Specifically, if $i^*$ could observe the other bids before committing funds, they could abort upon discovering an unfavorable price (e.g., $p_{i^*} \approx v_{i^*}$), forcing a restart.
By acting first, $i^*$ is effectively bound: an abort after this point results in forfeiture of $v_{i^*}$, which by individual rationality exceeds any possible payment $p_{i^*}(v)$.
Thus, the auctioneer need not restart; the forfeited deposit $v_{i^*}$ serves as sufficient revenue.

\begin{theorem}\label{thm:safety}
	The SRA satisfies safety against honest aborts.
\end{theorem}

\begin{proof}
	By \Cref{def:safety}, it suffices to show that upon any abort by a non-auctioneer participant, the protocol either (i) restarts with the remaining participants $P' \subset P$, preserving the original mechanism's outcome on $P'$, or (ii) terminates with revenue at least that of the mechanism on $P'$.
	We classify the possible abort conditions into three disjoint cases based on the protocol's progress:

	\begin{enumerate}
		\item \textbf{Abort during Computation (WDG):}
		      If the MPC aborts, the Identifiable Abort property ensures the deviation is attributed to a specific party $k$; if instead the gadget reports a party $k$ whose opening is inconsistent with its commitment, Step~\ref{step:exclude-inconsistent} names $k$ directly.
		      The protocol acts as if $k$ never participated.
		      The protocol restarts with $P \setminus \{k\}$, satisfying the definition.
		\item \textbf{Abort during Revelation (Ex-Ante):}
		      Suppose a bidder $j$ aborts \emph{before} the winner's deposit is secured.
		      \begin{itemize}
			      \item If $j$ is the winner $i^*$, they are excluded, and the protocol restarts.
			            The revenue equals the revenue with $P \setminus \{i^*\}$.
			      \item If $j$ is a loser, they are similarly excluded.
			            The allocation rule implies the winner remains $i^*$ (since $i^*$ outranks $j$), and the payment remains determined by the highest remaining losing bid.
			            Thus, the outcome is identical to the outcome with $P \setminus \{j\}$.
		      \end{itemize}
		\item \textbf{Abort during Revelation (Ex-Post):}
		      Suppose the winner $i^*$ aborts \emph{after} depositing $v_{i^*}$.
		      The protocol guarantees revenue $v_{i^*}$.
		      Since $v_{i^*} \ge p_{i^*}(v)$ (by individual rationality), an abort by $i^*$ can only increase the auctioneer's revenue.
	\end{enumerate}

	In all cases, the protocol recovers at least the revenue of the Myerson mechanism on the surviving set of bidders.
\end{proof}

\subsection{Characterizing Deviations}

We characterize the auctioneer's safe deviations in the SRA.
Unlike real bidders, whose valuations are drawn from known distributions $F_i$, shill bidders are entirely controlled by the auctioneer.
Crucially, the auctioneer must commit not only to each shill's bid $b_j$ but also to a distribution for it, registered in the setup of Algorithm~\ref{alg:credible_protocol} before the commit phase; the gadget then ranks shill $j$ by the monotone non-decreasing \emph{ironed virtual value transformation} $\ivv_j : \mathbb{R}_+ \to \mathbb{R}$ and the ironed bid $\ironedbid_j$ of that distribution.
This transformation is analogous to the ironed virtual value function for real bidders, but is chosen strategically.

The auctioneer's strategy thus consists of:
\begin{enumerate}
	\item Selecting a set of shill bidders $S$, each $j \in S$ with bid $b_j$ and committed virtual value transformation $\ivv_j(\cdot)$.
	\item When shill $j \in S$ is identified as the winner, deciding whether to \emph{reveal} $b_j$ or \emph{hide} (abort and restart without $j$).
	\item When shill $j \in S$ is not identified as the winner, deciding whether to \emph{reveal} $b_j$ or \emph{hide} (refuse to open $b_j$, forfeiting $\pen$; no restart occurs, and $j$ is dropped from the set on which the payment is computed).
\end{enumerate}

For credibility, we require that the honest mechanism weakly dominates \emph{all} such deviations---regardless of the choice of shill bids $\{b_j\}_{j \in S}$ and transformations $\{\ivv_j\}_{j \in S}$.

\begin{definition}[Threshold virtual value]
	\label{def:threshold}
	For bidding profile $\bm{b}$ and bidder $i \in N \cup S$, the \emph{threshold virtual value} is $\beta_i(\bm{b}) \coloneqq \max_{j \neq i} \ivv_j(b_j)$, the highest ironed virtual value among all other bidders.
\end{definition}

We analyze the auctioneer's decision by considering whether a shill is outbid (does not win) or wins the allocation.

\begin{observation}
	\label{obs:shill_wins_highest}
	Shill $j \in S$ is identified as the winner if and only if $\beta_j(\bm b) < \ivv_j(b_j)$ and $\ivv_j(b_j) \ge 0$, where $\bm b$ is the bidding profile of all non-aborting bidders.
\end{observation}

\begin{proof}
	The winner determination gadget implements the Myerson optimal allocation of \Cref{def:myerson_mech}, which allocates to the highest-ranked participant with non-negative ironed virtual value.
	Shill $j$ is therefore the winner exactly when it outranks every other non-aborting bidder, real or shill.
\end{proof}

\paragraph*{Refinements.}
Three assumptions on the auctioneer's deviations are without loss of generality.
First, we assume a shill bidder never aborts during the WDG computation phase since the protocol restarts without them.
This is strategically equivalent to the shill bidder aborting immediately after the WDG output is revealed (identifying them as the winner or loser).
Thus, we can focus our analysis on deviations that occur during the revelation phase.
Second, we assume every shill has non-negative ironed virtual value, $\ivv_j(b_j) \ge 0$: a shill with $\ivv_j(b_j) < 0$ never wins (\Cref{obs:shill_wins_highest}) and never raises the winner's payment, since any bid clearing the reserve, $\ivv_i(z) \ge 0$, already outranks it, so removing it changes neither the outcome nor the penalties.
Third, we assume an abort names the provisional winner: the ideal adversary of \Cref{def:ideal_ia} may name any corrupted party, but aborting a shill other than the provisional winner is weakly dominated, since the re-run returns the same provisional winner without providing any information advantage.

The next lemma shows that if a shill is outbid by a real bidder, then revealing dominates hiding.

\begin{lemma}[Outbid shills]
	\label{lem:reveal_outbid}
	For shill $j \in S$, if $\ivv_j(b_j) < \beta_j(\bm b)$, then revealing $b_j$ weakly dominates aborting.
\end{lemma}
\begin{proof}
	By \Cref{obs:shill_wins_highest}, if $\ivv_j(b_j) < \beta_j(\bm b)$ then shill $j$ is not the winner, so some other bidder $i^* \neq j$ wins and reveals first.
	At this point revealing $b_j$ acts as a \emph{price floor}: the winner pays the least bid beating $\beta_{i^*}(\bm b)$, and including $b_j$ among the competitors can only raise that maximum, hence it can only raise the payment.
	Aborting instead removes $b_j$ from the competitor set, weakly decreasing revenue while incurring penalty $\pen$.
\end{proof}

\begin{definition}[Refined safe deviation]
	\label{def:refined_deviation}
	A \emph{refined safe deviation} is a safe deviation in which, when choosing whether to reveal or abort $j \in S$, the deviation reveals $b_j$ if any bid (shill or real) has been revealed so far, and otherwise chooses freely.
\end{definition}

\begin{lemma}[Refinement]
	\label{lem:refinement}
	For any safe deviation there exists a refined safe deviation with weakly greater expected revenue.
	Moreover, the shills of any safe deviation may be relabelled so that, writing $r_j$ for the bid of the $j$-th shill and $\ivv(r_j)$ for $\ivv_j(r_j)$, its ironed virtual value under its own committed transformation, $\ivv(r_1) < \ivv(r_2) < \cdots < \ivv(r_k)$ in the ranking of \Cref{def:ranking}.
\end{lemma}

\begin{proof}
	The ordering is immediate: the ranking is a strict total order (\Cref{def:ranking}), so it linearly orders any finite set of shills whatever transformations $\{\ivv_j\}_{j \in S}$ they were committed to.

	For the reveal-after-reveal property, suppose some bid has been revealed.
	The winner is then already determined, so every remaining shill $j$ satisfies $\ivv_j(b_j) < \beta_j(\bm b)$, and \Cref{lem:reveal_outbid} says revealing $b_j$ weakly dominates aborting.
\end{proof}

\subsection{DSIC Property of the SRA}

We now establish that when all shills are revealed, the SRA reduces to a DSIC mechanism.

\begin{theorem}[DSIC under full revelation]
	\label{thm:dsic_full_reveal}
	For a refined safe deviation of the Sequential Revelation Auction, if the auctioneer reveals all shills (i.e., never aborts), then this deviation is Dominant Strategy Incentive Compatible (DSIC) for the real bidders.
\end{theorem}

\begin{proof}
	Under full reveal no participant is excluded, so the allocation is exactly that of \Cref{def:myerson_mech} applied to $A = N \cup S$: the item goes to the highest-ranked participant with non-negative ironed virtual value, and the winner pays the least bid that would still beat $\beta_i(\bm b)$.
	The allocation is monotone in each real bidder's own bid by the monotonicity of \Cref{def:ranking}, and its payments are the induced threshold payments, so \Cref{thm:myerson_char} gives that the deviation is DSIC for the real bidders.
\end{proof}

We close this section by verifying that honest behavior is also optimal for the real bidders in the \emph{extensive form}, not merely in the reported values.
Two steps of Algorithm~\ref{alg:credible_protocol} warrant this check: the provisional winner must escrow the full $v_{i^*}$ before learning the price, and the losers must open commitments that serve only to raise the winner's payment.

\begin{proposition}[Honest play is weakly dominant for real bidders]
	\label{prop:bidder_ir}
	Fix $\pen > 0$ and a real bidder $i \in N$ participating in the SRA.
	Conditional on reaching any information set of Algorithm~\ref{alg:credible_protocol}, following the protocol weakly dominates aborting.
\end{proposition}
\begin{proof}
	A bidder who aborts forfeits the collateral $\pen$ and is excluded from the continuation, so it suffices to show that continuing yields a non-negative payoff.
	There are two cases.

	If $i$ is a loser, then $i$ receives no allocation and makes no payment whether or not they open their commitment: their payoff from continuing is $0$, while aborting costs $\pen > 0$.
	(Real bidders are not revenue-maximizers, so the fact that opening raises the winner's payment does not enter $i$'s utility.)

	If $i$ is the provisional winner, continuing means escrowing $v_i$ and later receiving the refund $v_i - p_i(\bm{v}_A)$, for a net payoff of $v_i - p_i(\bm{v}_A) \ge 0$ by individual rationality of the Myerson payment rule: $p_i(\bm{v}_A)$ is the least bid beating $\beta_i(\bm{v}_A)$, which is at most $v_i$ whenever $i$ wins.
	Aborting instead forfeits $\pen$ and yields no allocation, for a payoff of $-\pen < 0$.
\end{proof}

\section{Example: Equal Revenue Distribution}\label{sec:example}

We now illustrate credibility for the canonical equal revenue distribution $F(v) = 1 - 1/v$ ($v \geq 1$).
This distribution serves as a critical stress test for credibility due to two distinct properties:
\begin{enumerate}
	\item \textbf{Heavy Tail:}
	      The polynomial tail $\pr{v > t} = 1/t$ implies that high valuations occur with probability decaying only polynomially.
	      This allows shill-bidding strategies to chase huge payouts without the risk vanishing exponentially fast.
	\item \textbf{Constant Revenue:}
	      The ``equal revenue'' property ($r \cdot \pr{v \geq r} = 1$) implies that every reserve price yields the identical expected revenue of $1$.
	      The distribution is ``flat'' with respect to optimization.
\end{enumerate}
\citet{FerreiraWeinberg} demonstrate that for deferred revelation auctions (without MPC gadgets), \emph{no finite penalty} suffices for this distribution---the heavy tail allows geometric shill strategies to extract unbounded profit.
As a warmup, we show how a penalty of $\pen \ge 1$ prevents unbounded profit for a natural class of shill strategies under the SRA.
In \Cref{sec:lower}, we prove this result is tight.

\begin{example}[Warmup: Equal Revenue Distribution]
	\label{ex:equal_revenue_warmup}
	Consider a single real bidder with valuation drawn from the equal revenue distribution $F(v) = 1 - 1/v$ for $v \geq 1$.
	Suppose the auctioneer employs a safe deviation using $k$ shill bids $r_1 < r_2 < \cdots < r_k$, where they hide $r_i, \ldots, r_k$ if and only if $v \in [r_{i-1}, r_i)$.
	With a penalty $\pen \geq \rev(F) = 1$, we show this deviation yields expected revenue at most $1$.

	Without loss of generality, $r_j \geq 1$ for all $j$ (shills below the support never win).
	Set $r_0 = 1$ (the lower bound of the support) and $r_{k+1} = \infty$.
	For the equal revenue distribution, the mechanism allocates to the highest bidder (any reserve $r \geq 1$ is optimal).

	When $v \in [r_{i-1}, r_i)$ the auctioneer hides $r_i, \ldots, r_k$ (each incurring penalty $\pen$), while shills $r_1, \ldots, r_{i-1}$ are outbid and revealed as price floors (by \Cref{lem:reveal_outbid}).
	The bidder wins and pays $r_{i-1}$, giving ex-post revenue: \[ R(v) = r_{i-1} - (k - i + 1) \cdot \pen.
	\]

	We decompose the expected revenue into gross payments minus penalties.
	Using $\pr{v \in [r_{i-1}, r_i)} = 1/r_{i-1} - 1/r_i$, the gross payment term telescopes: \[ \text{Gross} = \sum_{i=1}^{k+1}\left(\frac{1}{r_{i-1}} - \frac{1}{r_i}\right) r_{i-1} = \sum_{i=1}^{k+1}\left(1 - \frac{r_{i-1}}{r_i}\right).
	\]
	For the penalties, each shill $r_j$ is penalized whenever it \emph{wins}, i.e., when $v < r_j$, which occurs with probability $1 - 1/r_j$.
	Thus the total expected penalty is $\pen \sum_{j=1}^{k}(1 - 1/r_j)$.
	Since the penalty term is non-decreasing in $\pen$, it suffices to take $\pen = 1$; separating the last term of the gross sum ($1 - r_k/r_{k+1} = 1$),
	\begin{align*}
		\e{R(v)} & = 1 + \sum_{j=1}^{k}\left[\underbrace{\left(1 - \frac{r_{j-1}}{r_j}\right)}_{\text{gross from interval } j} - \underbrace{\left(1 - \frac{1}{r_j}\right)}_{\text{penalty for shill } r_j}\right] = 1 + \sum_{j=1}^{k} \frac{1 - r_{j-1}}{r_j}.
	\end{align*}
	Since $r_{j-1} \geq r_0 = 1$ for all $j \geq 1$, each summand is non-positive, giving: \[ \e{R(v)} \leq 1 = \rev(F).
	\]
\end{example}

\begin{standalonebib}
	\bibliographystyle{ACM-Reference-Format}
	\bibliography{mybib}
\end{standalonebib}

\section{Bounding Information Leakage Under Abort}\label{sec:info_leakage}

By Myerson's characterization, the expected revenue of any DSIC mechanism equals its expected virtual surplus.
However, this classical equivalence holds only when the expectation is taken over the full, unconditioned prior distribution.
In our credibility analysis, the auctioneer's decision to abort and restart the auction hinges on a conditional event---specifically, learning that the highest real bid is below a certain threshold.
When we condition on such an upper-bound event---the event $E_r$ of \Cref{lem:information_leakage} below, on which no real bidder outranks the shill $r$---the boundary terms in Myerson's integration by parts no longer vanish, and expected revenue can exceed expected virtual surplus.

The \emph{Information Leakage Lemma} quantifies precisely how much the conditional expected revenue can exceed the conditional expected virtual surplus.
Crucially, we prove that this excess ``leakage'' is bounded by $\sum_{i=1}^n \rev(F_i)$ (the sum of the optimal monopoly revenues for each bidder).
This bounds the maximum option value the auctioneer can extract from restarting, dictating the minimum penalty $\pen$ required to deter deviations.

\paragraph*{Intuition via a single bidder.}
Consider a single bidder with valuation $v \sim F$, reserve price $r^*$, allocation rule $x(v) = \ind{v \ge r^*}$, and payment rule $p(v) = r^* \ind{v \ge r^*}$.
Under DSIC, the utility of a bidder of type $v$ is given by the envelope formula $u(v) = \int_0^v x(t)\, dt$, which represents the bidder's accumulated information rent.
Myerson's identity shows that the unconditional expected revenue is $\rev(F) = r^* (1 - F(r^*)) = \e{\phi(v)^+}$, where $\phi(v)$ is the virtual value.
Now, condition on the upper-bound event $E = \{v \le c\}$ for some threshold $c > r^*$.
Working with unnormalized expectations throughout, the truncated expected payment is $\e{p(v) \ind{v \le c}} = \e{\left(v x(v) - \int_0^v x(t)\, dt\right) \ind{v \le c}}$.
When we perform integration by parts on the truncated information rent term $\int_0^c f(v) \int_0^v x(t)\, dt\, dv$, the boundary term at $v = c$ does not vanish, yielding: \[ \int_0^c f(v) \int_0^v x(t)\, dt\, dv = F(c) \int_0^c x(t)\, dt - \int_0^c F(v) x(v)\, dv.
\]
Rearranging terms, the truncated expected revenue satisfies:
\[ \e{p(v) \ind{v \le c}} = \e{\phi(v) x(v) \ind{v \le c}} + (1-F(c)) \int_0^c x(t)\, dt. \]
The second term on the right-hand side is the non-vanishing boundary term representing the information leakage: the utility $u(c) = \int_0^c x(t)\, dt = c - r^*$ of the boundary type $c$ weighted by the probability of exceeding $c$ in the unconditioned world.
Crucially, this leakage is bounded by \[ (1 - F(c))\,(c - r^*) \;\le\; c\,(1 - F(c)) \;\le\; \rev(F), \] where the last inequality holds because posting the price $c$ cannot beat the optimal monopoly price.
The Information Leakage Lemma generalizes this analysis to $n$ bidders.
With several bidders, the same integration by parts produces the virtual value $\phi_i$ of each bidder on a truncated domain, and the proof replaces it there by the ironed virtual value $\ivv_i$, which is the content of the following identity.

\begin{lemma}[Ironing Equivalence; {\citet{Myerson81}}]\label{lem:ironing_equivalence}
	Let $F_i$ be a distribution with virtual value function $\phi_i$ and ironed virtual value function $\ivv_i$.
	Let $x_i$ be an allocation rule such that $x_i(\cdot, \bm{v}_{-i})$ is constant on every ironed interval of $F_i$.
	Let $c$ be any threshold that does not lie in the strict interior of an ironed interval of $F_i$.
	Then for any fixed $\bm{v}_{-i}$, $$\int_0^c x_i(v_i, \bm{v}_{-i})\, \phi_i(v_i)\, f_i(v_i)\, dv_i = \int_0^c x_i(v_i, \bm{v}_{-i})\, \ivv_i(v_i)\, f_i(v_i)\, dv_i.
	$$
\end{lemma}

This identity holds because $x_i$ is constant on any ironed interval $I = [\underline{v}, \overline{v}]$ of $F_i$ by hypothesis.
Thus, on each complete ironed interval, the allocation rule $x_i(v_i, \bm{v}_{-i})$ can be factored out of the integral.
Since $\ivv_i$ is defined as the conditional expectation of $\phi_i$ with respect to the distribution $F_i$ on that interval, we have $\int_I \phi_i(v_i) f_i(v_i)\, dv_i = \ivv_i(v) \int_I f_i(v_i)\, dv_i$ for any $v \in I$.
Because the threshold $c$ never bisects an ironed interval, the domain $[0, c]$ is partitioned into complete ironed intervals and regions where $\phi_i = \ivv_i$.
Thus, the equivalence holds over the truncated domain; see \citet[Section~4]{Myerson81}.

We state the lemma for the class of mechanisms that arise as continuations of the SRA after some shills have been revealed: the Myerson optimal auction augmented with a price floor.
Formally, for a shill bid $r \ge 0$ let $M^{(r)}$ denote the DSIC mechanism with allocation rule \[ x_i(\bm{v}) = \ind{\ivv_j(v_j) < \ivv_i(v_i) \ \forall j \in N \setminus \{i\}} \cdot \ind{\ivv(r) < \ivv_i(v_i)} \cdot \ind{\ivv_i(v_i) \ge 0} \] and the induced threshold payments, all comparisons being taken in the ranking of \Cref{def:ranking}.
Placing $r$ below every participant recovers the Myerson optimal auction, and $M^{(r_i)}$ is the continuation in which the shills above $r_i$ have been aborted and $r_1 < \cdots < r_i$ remain, to be revealed as price floors.
Write $R^{(r)}(\bm v) = \sum_{i=1}^n p_i(\bm v)$ and $\vs^{(r)}(\bm v) = \sum_{i=1}^n x_i(\bm v)\, \ivv_i(v_i)$ for the ex-post revenue and the ex-post ironed virtual surplus of $M^{(r)}$.

\begin{lemma}[Information Leakage Lemma]\label{lem:information_leakage}
	Let there be $n$ bidders with independent valuations drawn from distributions $F_1, \dots, F_n$ with vanishing revenue tails, $\lim_{v \to \infty} v\,(1 - F_i(v)) = 0$.
	For a shill bid $r$ with $\ivv(r) \ge 0$, define the event \[ E_r = \{\bm{v} \mid \forall i,\, \ivv_i(v_i) \le \ivv(r)\} \] on which no real bidder outranks $r$.
	Let $R(\bm{v}) = \sum_{i=1}^n p_i(\bm{v})$ and $\vs(\bm{v}) = \sum_{i=1}^n x_i(\bm{v}) \ivv_i(v_i)$ denote the ex-post revenue and ex-post ironed virtual surplus, respectively, of $M^{(r')}$ for any floor $r' \le r$.
	Then, \[\e{R(\bm{v}) \ind{E_r}} \le \e{\vs(\bm{v}) \ind{E_r}} + M \cdot \pr{E_r}\] where $M = \sum_{i=1}^n \rev(F_i)$ and $\rev(F_i)$ is the expected revenue of the optimal monopoly posted price for distribution $F_i$.
\end{lemma}

Observe the event $E_r$ is a product of half-lines, $E_r = \prod_{i \in N} \{v_i < c_i\}$, where $c_i := \inf\{z : \ivv(r) < \ivv_i(z)\}$ is the least bid at which $i$ would overtake $r$ in the ranking of \Cref{def:ranking}; the set of bids at which $i$ outranks $r$ is an up-set by the monotonicity of the ranking, which is what makes each factor a half-line.
This product form is the only property of $E_r$ the proof uses.
Note that $c_i$ need not be $\ivv_i^{-1}(\ivv(r))$: when $\ivv_i$ is flat at the level $\ivv(r)$, the tie is settled on the ironed bid, so $c_i$ is $r$ on a flat piece of a regular $\phi_i$ and an endpoint of the ironed interval at that level otherwise.
In neither case does $c_i$ lie in the interior of an ironed interval of $F_i$: on an ironed interval the rank of $i$ is constant, so $i$ outranks $r$ either throughout it or nowhere in it, and the infimum $c_i$ cannot fall strictly inside.

\submission{The proof of \Cref{lem:information_leakage} is deferred to the full version.}
\arxiv{
	\begin{proof}
		Recall the thresholds $c_i$ and the product form $E_r = \{\bm{v} : v_i < c_i \text{ for all } i\}$.
		If $\pr{E_r} = 0$ both sides of the inequality vanish, so assume $F_i(c_i) > 0$ for every $i$.
		We write the integrals below over $[0, c_i]$; the endpoint carries no mass, since $F_i$ has a density.
		We prove the lemma in two steps: an exact identity, \eqref{eq:revenue_identity}, expressing the truncated revenue as the truncated ironed virtual surplus plus one boundary term per bidder, and a bound of each boundary term by $\rev(F_i)$.

		Fix a bidder $i$ and a profile $\bm{v}_{-i}$ with $v_j < c_j$ for all $j \neq i$.
		By Myerson's characterization (\Cref{thm:myerson_char}), the payment of $i$ in the DSIC mechanism $M^{(r')}$ is $p_i(\bm{v}) = v_i x_i(\bm{v}) - \int_0^{v_i} x_i(t, \bm{v}_{-i})\, dt$, so its expectation truncated at $v_i \le c_i$ is
		$$\int_0^{c_i} p_i(v_i, \bm{v}_{-i}) f_i(v_i)\, dv_i = \int_0^{c_i} \left( v_i x_i(v_i, \bm{v}_{-i}) - \int_0^{v_i} x_i(t, \bm{v}_{-i})\, dt \right) f_i(v_i)\, dv_i.$$
		Integrating the double integral by parts, with $-(F_i(c_i) - F_i(v_i))$ as the antiderivative of $f_i(v_i)$, gives
		$$\int_0^{c_i} p_i(\bm{v}) f_i(v_i)\, dv_i = \int_0^{c_i} x_i(\bm{v}) \left( v_i - \frac{1 - F_i(v_i)}{f_i(v_i)} + \frac{1 - F_i(c_i)}{f_i(v_i)} \right) f_i(v_i)\, dv_i,$$
		where the bracket is the virtual value $\phi_i(v_i)$ plus the boundary correction $(1 - F_i(c_i))/f_i(v_i)$.
		Next, we replace $\phi_i$ by $\ivv_i$ under the allocation rule.
		The Ironing Equivalence (\Cref{lem:ironing_equivalence}) applies once $x_i(\cdot, \bm{v}_{-i})$ is constant on every ironed interval of $F_i$ and $c_i$ lies outside the interior of every ironed interval; we check both.
		Note $x_i(\cdot, \bm{v}_{-i})$ depends on $v_i$ only through the rank of $i$, which \Cref{def:ranking} keeps constant on each ironed interval of $F_i$.
		Note $c_i$ is never interior to an ironed interval of $F_i$, as observed after the statement.
		Thus the Ironing Equivalence gives
		$$\int_0^{c_i} p_i(\bm{v}) f_i(v_i)\, dv_i = \int_0^{c_i} x_i(\bm{v}) \ivv_i(v_i) f_i(v_i)\, dv_i + (1 - F_i(c_i)) \int_0^{c_i} x_i(v_i, \bm{v}_{-i})\, dv_i.$$
		The argument so far uses only two properties of $M^{(r')}$: that it ranks by \Cref{def:ranking}, so its allocation is constant on ironed intervals, and that it is DSIC with threshold payments.
		Neither singles out the Myerson optimal auction, which is why the lemma covers every continuation mechanism arising in \Cref{sec:credible}.
		The last integral is $u_i(c_i, \bm{v}_{-i}; c_i)$, the ex-post utility of type $c_i$ when the other bidders report $\bm{v}_{-i}$.
		Taking the expectation over $\bm{v}_{-i}$ subject to $v_j < c_j$ for all $j \neq i$ and summing over the $n$ bidders gives the identity
		\begin{equation}\label{eq:revenue_identity}
			\e{R(\bm{v}) \ind{E_r}} = \e{\vs(\bm{v}) \ind{E_r}} + \sum_{i=1}^n (1 - F_i(c_i))\, \e{u_i(c_i, \bm{v}_{-i}; c_i) \ind{\bm{v}_{-i} \in E_r}},
		\end{equation}
		where $\bm{v}_{-i} \in E_r$ abbreviates $v_j < c_j$ for all $j \neq i$.

		Next, we bound the boundary terms.
		Dividing \eqref{eq:revenue_identity} by $\pr{E_r} = \prod_{j=1}^n F_j(c_j)$, it suffices to show $\sum_{i=1}^n M_i \le \sum_{i=1}^n \rev(F_i)$, where
		$$M_i := \frac{1 - F_i(c_i)}{F_i(c_i)}\, \e{u_i(c_i, \bm{v}_{-i}; c_i) \mid \bm{v}_{-i} \in E_r}.$$
		Write $r_i^* := \ivv_i^{-1}(0)$ for the least bid of $i$ with non-negative ironed virtual value, which is the monopoly price of $F_i$, so that $\rev(F_i) = r_i^*(1 - F_i(r_i^*))$.
		In $M^{(r')}$ the allocation requires $\ivv_i(v_i) \ge 0$, so $x_i(t, \bm{v}_{-i}) = 0$ for every $t < r_i^*$; a higher floor only lowers the utility bounded below, so it suffices to argue with $r_i^*$.
		Note $c_i \ge \ivv_i^{-1}(\ivv(r)) \ge \ivv_i^{-1}(0) = r_i^*$, since no bid $z$ with $\ivv_i(z) < \ivv(r)$ outranks $r$ and $\ivv(r) \ge 0$.
		Thus the utility of type $c_i$ is at most the margin above the reserve, $u_i(c_i, \bm{v}_{-i}; c_i) = \int_{r_i^*}^{c_i} x_i(t, \bm{v}_{-i})\, dt \le c_i - r_i^*$, and
		$$M_i \le \frac{(1 - F_i(c_i))(c_i - r_i^*)}{F_i(c_i)} = \frac{c_i(1 - F_i(c_i)) - r_i^*(1 - F_i(c_i))}{F_i(c_i)}.$$
		Note $c_i (1 - F_i(c_i)) \le \rev(F_i)$, since posting the price $c_i$ earns at most the monopoly revenue, and $\rev(F_i) = r_i^*(1 - F_i(r_i^*)) \le r_i^*$.
		Thus, substituting the first inequality in the first term of the numerator and the second in its subtracted term,
		$$M_i \le \frac{\rev(F_i) - \rev(F_i)(1 - F_i(c_i))}{F_i(c_i)} = \rev(F_i).$$
		Summing over $i$ gives $\sum_{i=1}^n M_i \le \sum_{i=1}^n \rev(F_i)$, and multiplying \eqref{eq:revenue_identity} back by $\pr{E_r}$ proves the lemma.
	\end{proof}
}

\section{Sequential Revelation Auction is Credible}\label{sec:credible}

The warmup (\Cref{sec:example}) demonstrates credibility for a restricted class of strategies against the equal-revenue distribution.
The Information Leakage Lemma (\Cref{sec:info_leakage}) establishes the quantitative tool: the excess revenue an auctioneer can extract from conditional information is bounded by $\sum_{i=1}^n \rev(F_i)$.
We combine these ingredients to prove credibility in full generality.

\begin{theorem}[Credibility of the sequential revelation auction]\label{thm:credibility}
	Consider a single-item auction with $n$ bidders with independent valuations drawn from product distribution $\bm{F}$ whose components have densities and vanishing revenue tails, $\lim_{v \to \infty} v\,(1 - F_i(v)) = 0$.
	Then the Sequential Revelation Auction (Algorithm~\ref{alg:credible_protocol}) with penalty $\pen \geq \sum_{i=1}^n \rev(F_i)$ is credible.
\end{theorem}

\paragraph*{Proof roadmap.}
The proof reduces the auctioneer's complex, adaptive strategy space to a finite number of deterministic cases, then bounds each one.
Throughout, $M^{(r)}$ is the continuation mechanism of \Cref{sec:info_leakage}, the Myerson optimal auction with a price floor at $r$, and $R^{(r)}$ and $\vs^{(r)}$ are its ex-post revenue and ironed virtual surplus.
The reduction has three layers:
\begin{enumerate}
	\item \textbf{Refinement} (\Cref{lem:refinement}): Without loss of revenue, the shills are linearly ordered by the ranking of \Cref{def:ranking} and outbid shills are always revealed.
	\item \textbf{Kuhn's Theorem} (\Cref{lem:reduction_pure}): Any adaptive (behavioral) strategy is equivalent to a mixture of pure strategies, each of which commits to a threshold: ``abort all shills above $r_i$, reveal $r_i$ and below.''
	\item \textbf{Case analysis} (\Cref{lem:unified_cases}): For each threshold $r_i$, the expected net revenue decomposes over the shill at which the descent stops.
	      Each abort of a shill $r_j$ leaks at most $\sum_{i'} \rev(F_{i'}) \cdot \pr{\mathcal{E}_j}$ of value by the Information Leakage Lemma and costs $\pen \cdot \pr{\mathcal{E}_j}$ in forfeited collateral, so the penalty $\pen \ge \sum_{i'} \rev(F_{i'})$ offsets it abort by abort.
\end{enumerate}
The central tension is a \emph{commitment \versus free option} trade-off.
Revealing a shill commits the auctioneer to a fixed price floor, whereas hiding it opens a ``free option'' to restart---but at cost $\pen$.
We show the option value is always bounded by $\sum_j \rev(F_j)$, so the penalty neutralizes it.

We start by bounding the revenue on good events for the auctioneer---i.e., when a real bidder is guaranteed to be the winner.
We show that on the event that a real bidder outranks the floor, the expected revenue of $M^{(r)}$ equals the expected maximum ironed virtual value on that event.

\begin{lemma}[Revenue on good events]\label{lem:revenue_E}
	Let $r$ be a shill bid with $\ivv(r) \ge 0$, and let $E_r^c = \{ \bm{v} : \exists i \in N,\, \ivv(r) < \ivv_i(v_i) \}$ be the event that some real bidder outranks $r$.
	Then \[ \e{R^{(r)}(\bm v) \cdot \ind{E_r^c}} = \e{\max_{i \in N} \ivv_i(v_i) \cdot \ind{E_r^c}}.
	\]
\end{lemma}
\submission{The proof of \Cref{lem:revenue_E} is deferred to the full version.}
\arxiv{
	\begin{proof}
		Since the ranking is monotone in a participant's own bid (\Cref{def:ranking}), the allocation rule of $M^{(r)}$ is non-decreasing in each $v_i$, and its payments are the induced threshold payments, so $M^{(r)}$ is DSIC by \Cref{thm:myerson_char}.
		Note the allocation $x_i(\cdot, \bm v_{-i})$ of $M^{(r)}$ depends on $v_i$ only through the rank of $i$, which is constant on every ironed interval of $F_i$.
		Thus \Cref{lem:myerson} applies with equality: \[ \e{R^{(r)}(\bm v)} = \e{\sum_{i \in N} \ivv_i(v_i) \cdot x_i(\bm{v})}.
		\]
		We now decompose both sides over $E_r^c$ and its complement.
		On $E_r^c$, some real bidder outranks $r$, so $M^{(r)}$ allocates to the highest-ranked real bidder, whose ironed virtual value is $\max_{i \in N} \ivv_i(v_i)$.
		On $E_r$, no real bidder outranks $r$, so $M^{(r)}$ allocates to no real bidder, no real bidder pays, and both the revenue and the virtual surplus vanish.
		Thus $\e{R^{(r)}(\bm v) \ind{E_r^c}} = \e{R^{(r)}(\bm v)} = \e{\max_{i \in N} \ivv_i(v_i) \ind{E_r^c}}$, as desired.
	\end{proof}
}

\subsection{Bounding Refined Deviations}

By \Cref{lem:refinement}, it suffices to consider refined safe deviations, where the auctioneer submits shill bids $r_1 < \dots < r_k$---ordered by the ranking of \Cref{def:ranking}---and processes them from highest to lowest.
We define $r_{k+1}$ to be a value that every participant precedes and $r_0$ to be a value that every participant outranks, so that $M^{(r_0)}$ is the Myerson mechanism of \Cref{def:myerson_mech} with no shill.

Let $\mathcal{E}_j$ be the event that no real bidder outranks the $j$-th shill: \[ \mathcal{E}_j := \{\bm{v}_N \in \mathbb R_+^N : \ivv_i(v_i) \le \ivv(r_j) \text{ for all } i \in N\}, \] with $\mathcal{E}_0 := \emptyset$ and $\mathcal{E}_{k+1} := \bm V$.
Because the ranking is a total order, $\mathcal{E}_j$ \emph{is} the event that shill $r_j$ is the provisional winner, by \Cref{obs:shill_wins_highest}---not merely a superset of it.
This is what the tie-break inside the ranking is for: had ties been left to a rule outside the comparison, a real bidder tied with $r_j$ but carrying a higher raw bid would lie in $\mathcal{E}_j$ while $r_j$ nonetheless loses, and the penalty accounting in \Cref{lem:unified_cases} would charge for an abort that never happens.
It is the event $E_{r_j}$ of the Information Leakage Lemma (\Cref{lem:information_leakage}) and the complement of the event of \Cref{lem:revenue_E}, which is what lets the lemmas compose in \Cref{lem:unified_cases}.
Since $r_1 < \cdots < r_k$, the events are nested: $\mathcal{E}_1 \subseteq \cdots \subseteq \mathcal{E}_k$.

Next, we formalize the deterministic strategies.
Any deterministic strategy corresponds to a case $C_i$ for some $i \in \{0, 1, \ldots, k\}$:
\begin{quote}
	$C_i$: \emph{The auctioneer reveals shill $r_i$ (if $i \ge 1$) and aborts all higher shills $r_{i+1}, \ldots, r_k$.
		The case $C_0$ corresponds to aborting all shills.
	}
\end{quote}
For a case $C_i$, write $R(\bm v; C_i)$ for the auctioneer's net revenue: the payments received from real bidders minus the collateral forfeited by aborted shills.

\begin{lemma}[Bounding Deterministic Cases]\label{lem:unified_cases}
	For any $i \in \{0, 1, \ldots, k\}$, the expected net revenue $\e{R(\bm v; C_i)}$ under $C_i$ is at most $\OPT(\bm{F})$, provided $\pen \ge \sum_{j=1}^n \rev(F_j)$.
\end{lemma}
\submission{The proof of \Cref{lem:unified_cases} is deferred to the full version.}
\arxiv{
	\begin{proof}
		For $\ell \in \{0, \ldots, k\}$ let $D_\ell := \mathcal{E}_{\ell+1} \setminus \mathcal{E}_\ell$ be the event that the highest-ranked real bidder outranks $r_\ell$ but not $r_{\ell+1}$; since the events $\mathcal{E}_j$ are nested, the $D_\ell$ partition $\bm V$.

		We first determine the outcome of $C_i$ on each slice.
		On $D_\ell$ with $\ell \ge i$, no real bidder outranks $r_{\ell+1}, \ldots, r_k$, so these shills win in turn and are aborted, costing $(k - \ell)\pen$; the gadget then runs on $N \cup \{r_1, \ldots, r_\ell\}$, some real bidder outranks $r_\ell$ and wins, and every surviving shill is outbid and revealed (\Cref{def:refined_deviation}), so the winner pays the least bid beating every other real bidder and $r_\ell$.
		The payments on $D_\ell$ are therefore those of $M^{(r_\ell)}$.
		On $\mathcal{E}_i = D_0 \cup \cdots \cup D_{i-1}$, no real bidder outranks $r_i$: the shills $r_k, \ldots, r_{i+1}$ are aborted, costing $(k-i)\pen$, and $r_i$ then wins and is revealed, so no real bidder pays.
		Thus \[ R(\bm v; C_i) = \sum_{\ell = i}^{k} R^{(r_\ell)}(\bm v) \ind{D_\ell} - \pen \sum_{j=i+1}^{k} \ind{\mathcal{E}_j}, \] where the penalty count is right because $\bm v \in \mathcal{E}_j$ exactly when $r_j$ is the provisional winner at its turn (\Cref{obs:shill_wins_highest}), and on $D_\ell$ this holds for the $k - \ell$ indices $j > \ell$, while on $\mathcal{E}_i$ it holds for all $k - i$ indices $j > i$.

		Next, we bound the payment terms.
		Since $R^{(r_\ell)}$ vanishes on $\mathcal{E}_\ell$, where no real bidder outranks the floor, $R^{(r_\ell)} \ind{D_\ell} = R^{(r_\ell)} \ind{\mathcal{E}_{\ell+1}}$.
		For $\ell \le k - 1$, the Information Leakage Lemma (\Cref{lem:information_leakage}) applies with the cut $r_{\ell+1}$, which satisfies $\ivv(r_{\ell+1}) \ge 0$, and the floor $r_\ell \le r_{\ell+1}$: \[ \e{R^{(r_\ell)}(\bm v) \ind{\mathcal{E}_{\ell+1}}} \le \e{\vs^{(r_\ell)}(\bm v) \ind{\mathcal{E}_{\ell+1}}} + M \cdot \pr{\mathcal{E}_{\ell+1}}, \] where $M = \sum_{j=1}^n \rev(F_j)$.
		For $\ell = k$, \Cref{lem:revenue_E} with the floor $r_k$ gives $\e{R^{(r_k)}(\bm v) \ind{\mathcal{E}_{k+1}}} = \e{R^{(r_k)}(\bm v)} = \e{\max_{i' \in N} \ivv_{i'}(v_{i'}) \ind{\mathcal{E}_k^c}}$, with no leakage term.
		On $D_\ell$ with $\ell \ge 1$ the winner of $M^{(r_\ell)}$ is the highest-ranked real bidder, whose ironed virtual value is at least $\ivv(r_\ell) \ge 0$; on $D_0$ the Myerson mechanism $M^{(r_0)}$ allocates to the highest-ranked real bidder exactly when its ironed virtual value is non-negative; and on $\mathcal{E}_\ell$ no real bidder wins.
		Thus $\vs^{(r_\ell)} \ind{\mathcal{E}_{\ell+1}} = \max_{i' \in N} \ivv_{i'}(v_{i'})^+ \cdot \ind{D_\ell}$ for every $\ell$, and the $\ell = k$ term has the same form since $\mathcal{E}_k^c = D_k$.
		Summing over $\ell \ge i$, \[ \sum_{\ell = i}^{k} \e{\vs^{(r_\ell)}(\bm v) \ind{\mathcal{E}_{\ell+1}}} = \e{\max_{i' \in N} \ivv_{i'}(v_{i'})^+ \cdot \ind{\mathcal{E}_i^c}} \le \e{\max_{i' \in N} \ivv_{i'}(v_{i'})^+} = \OPT(\bm F), \] the last equality being Myerson's characterization of the optimal revenue (\Cref{lem:myerson}).

		Combining the displays, \[ \e{R(\bm v; C_i)} \le \OPT(\bm F) + M \sum_{\ell=i}^{k-1} \pr{\mathcal{E}_{\ell+1}} - \pen \sum_{j=i+1}^{k} \pr{\mathcal{E}_j} = \OPT(\bm F) + (M - \pen) \sum_{j=i+1}^{k} \pr{\mathcal{E}_j}.
		\]
		Since $\pen \ge M$, the last sum is non-positive.
		This proves that the expected net revenue under $C_i$ is at most $\OPT(\bm F)$: each abort leaks at most $M \cdot \pr{\mathcal{E}_j}$ of value to the auctioneer and costs $\pen \cdot \pr{\mathcal{E}_j}$ in forfeited collateral.
	\end{proof}
}

Finally, we restrict the auctioneer's adaptive strategies to pure strategies.
In the extensive-form game, the auctioneer observes the history of aborted shills and may randomize at each decision; the reduction is Kuhn's Theorem.

\begin{theorem}[Kuhn's Theorem; {\citet{Kuhn53}}]\label{thm:kuhn}
	In any finite extensive-form game with perfect recall, for any behavioral strategy $\beta$, there exists a mixed strategy $\mu$ (a probability distribution over pure strategies) that is outcome-equivalent to $\beta$.
\end{theorem}

\begin{lemma}[Reduction to Pure Strategies]\label{lem:reduction_pure}
	The maximum expected revenue achievable by any behavioral strategy $\beta$ in the refined deviation space is bounded by the maximum expected revenue of the deterministic cases $\{C_0, \dots, C_k\}$.
\end{lemma}
\submission{The proof of \Cref{lem:reduction_pure} is deferred to the full version.}
\arxiv{
	\begin{proof}
		The auctioneer has perfect recall, and although the valuations are continuous, the auctioneer has only $k$ information sets---one per shill, reached when that shill is the provisional winner---each reached at most once along any play.
		Thus \Cref{thm:kuhn} applies to the auctioneer's decisions: a behavioral strategy $\beta$ that reveals the $j$-th shill with probability $q_j$ yields the same expected revenue as the mixed strategy $\mu$ that draws a reveal-set $T \subseteq \{1, \ldots, k\}$ with independent marginals $\pr{j \in T} = q_j$, since each coin is flipped at most once on every play.
		Next, we show that every pure strategy is outcome-equivalent to one of the $k+1$ cases $\{C_0, \ldots, C_k\}$.

		A pure strategy, a priori, is an arbitrary reveal-set $T \subseteq \{1, \ldots, k\}$: for each shill $r_j$ it fixes in advance whether that shill would be revealed or aborted if it were reached.
		There are $2^k$ such sets, not $k+1$, so this step requires argument.
		Fix a valuation profile $\bm v$ and let $\ell(\bm v)$ be the largest index with $\ivv(r_\ell) \le \max_{i \in N} \ivv_i(v_i)$, i.e.\ the highest shill that a real bidder outranks (set $\ell = 0$ if no shill is outranked).
		Because the shills are processed from highest to lowest, the auctioneer under $T$ aborts shills $r_k, r_{k-1}, \ldots$ in turn, and the first index it \emph{reveals} is \[ m(\bm v) \;=\; \max\bigl(T \cap \{\ell(\bm v)+1, \ldots, k\}\bigr), \] if this set is non-empty; the descent stops there, shill $r_{m(\bm v)}$ takes the item, no real bidder pays, and lower shills are never reached as provisional winners.
		If the set is empty, every aborted shill lies above $\ell(\bm v)$, the descent passes all of them, and the highest-ranked real bidder wins (or nobody does, when $\ell(\bm v) = 0$ and no real bidder clears the reserve)---exactly the outcome of $C_{\ell(\bm v)}$ restricted to this profile.
		In either case the ex-post outcome under $T$ coincides with the outcome under the single case $C_{\max T}$: on profiles where $\max T > \ell(\bm v)$ the stopping index is $m(\bm v) = \max T$ and the two agree by construction, and on the remaining profiles both let the highest-ranked real bidder win at the floor $r_{\ell(\bm v)}$.
		Thus $T$ is outcome-equivalent to $C_{\max T}$ (with $C_0$ for $T = \emptyset$), so the mixture $\mu$ over pure strategies induces a distribution over $\{C_0, \ldots, C_k\}$.
		Therefore \[ \e{R(\bm{v}; \beta)} = \e{R(\bm{v}; \mu)} = \sum_{i=0}^k \mu(C_i)\, \e{R(\bm{v}; C_i)} \le \max_{i}\, \e{R(\bm{v}; C_i)}, \] so the expected revenue of $\beta$ is at most that of the best deterministic case, as desired.
	\end{proof}
}

\submission{The proof of \Cref{thm:credibility} is deferred to the full version.}
\arxiv{
	\begin{proof}[Proof of \Cref{thm:credibility}]
		By \Cref{lem:reduction_pure}, the expected revenue of any behavioral safe deviation $\beta$ is at most $\max_i \e{R(\bm{v}; C_i)}$.
		By \Cref{lem:unified_cases}, for all $i \in \{0, \ldots, k\}$, we have $\e{R(\bm{v}; C_i)} \le \OPT(\bm{F})$.
		Thus, no safe deviation is strictly profitable, and the SRA is credible.
	\end{proof}
}

\begin{theorem}\label{thm:main}
	For any product distribution $\bm{F}$ whose components have densities and vanishing revenue tails, the Sequential Revelation Auction (Algorithm~\ref{alg:credible_protocol}) with penalty $\pen \ge \sum_{i=1}^n \rev(F_i)$ is a credible auction protocol that implements the revenue-optimal auction.
	Moreover, on the equilibrium path without shills, the protocol terminates in a constant number of rounds.
\end{theorem}
\submission{The proof of \Cref{thm:main} is deferred to the full version.}
\arxiv{
	\begin{proof}
		By \Cref{thm:credibility}, the SRA is credible for any product distribution $\bm{F}$ when $\pen \ge \sum_{i=1}^n \rev(F_i)$.
		The underlying mechanism is that of \Cref{def:myerson_mech}, which is DSIC by \Cref{thm:myerson_char} and, since its allocation is constant on every ironed interval, revenue-optimal by \Cref{lem:myerson}.
		Regarding round complexity, the protocol consists of a sequence of MPC executions.
		On the equilibrium path (no shills, or no aborts by the winner), the first MPC execution (Winner Determination Gadget) identifies the highest-ranked bidder.
		This bidder then reveals their bid, and the auction terminates.
		The WDG can be implemented in a constant number of rounds using the protocol from \Cref{thm:mpc_rounds}.
		The reveal phase takes a constant number of rounds: the winner reveals and escrows, then the losers reveal, then the auctioneer refunds.
		Thus, the total round complexity is constant.
	\end{proof}
}

\section{Lower Bound}\label{sec:lower}

The following result shows that \Cref{thm:credibility} is tight: the Sequential Revelation Auction (Algorithm~\ref{alg:credible_protocol}) admits profitable safe deviations when the penalty $\pen$ is smaller than $\sum_{i=1}^n \rev(F_i)$.
The equal-revenue distribution itself, $F(v) = 1 - 1/v$ on $[1, \infty)$, has $v\,(1 - F(v)) = 1$ for every $v$ and so lies outside the hypothesis of \Cref{thm:credibility}; we therefore truncate it.
For $H \ge 1$, let $F_H$ be the distribution on $[1, \infty)$ with \[ 1 - F_H(v) = \begin{cases} 1/v & \text{if } 1 \le v \le H, \\ (1/H)\, e^{-(v - H)/H} & \text{if } v > H. \end{cases} \]
Its density is $f_H(v) = 1/v^2$ on $[1, H]$ and $f_H(v) = (1/H^2)\, e^{-(v - H)/H}$ beyond, continuous at $H$, and its virtual value is $\phi_H(v) = 0$ on $[1, H]$ and $\phi_H(v) = v - H$ beyond, so $F_H$ is regular with a vanishing revenue tail and satisfies every hypothesis of \Cref{thm:credibility}.
Posting any price $p \le H$ earns $p \cdot (1/p) = 1$, and $p\,(1 - F_H(p))$ decreases for $p > H$, so $\rev(F_H) = 1$; and $\pr{v \ge y} = 1/y$ for every $y \le H$, which is all the construction below uses.

\begin{theorem}[Profitable Safe Deviation of the SRA]
	\label{thm:lower_bound}
	Fix $n \ge 1$ and a penalty $\pen < n$.
	There exists $H_0$ such that, for every $H \ge H_0$, the Sequential Revelation Auction with $n$ bidders whose valuations are drawn \iid from $F_H$ admits a safe deviation with expected net revenue strictly greater than $n$.
	Since $\OPT(\bm F) \le \sum_{i=1}^n \rev(F_H) = n$ for $\bm F = (F_H, \ldots, F_H)$, every such deviation is profitable, and the penalty $\pen \ge \sum_{i=1}^n \rev(F_i)$ of \Cref{thm:main} is exactly tight on this family.
\end{theorem}

The deviation places shill bids at exponentially growing values $r_j = \alpha^j$ for a constant $\alpha > 1$ fixed in the proof.
Each additional shill $r_j$ ``covers'' the interval $[r_{j-1}, r_j)$: conditional on the highest real valuation $V$ falling in this interval, the shill at $r_j$ wins (incurring penalty $\pen$), but the previous shill $r_{j-1}$ sets the price, extracting revenue $r_{j-1}$.
Below $H$, the probability mass of this interval decays at rate $1/r_{j-1}$, balancing the price $r_{j-1}$ to yield a constant expected gain, whereas the penalty cost is constant.
If the gain exceeds the cost (i.e., $\pen < n$), the auctioneer can extract revenue that grows without bound in the number of shills.

\submission{
	We defer the full proof to the full version.
}
\arxiv{
	\begin{proof}
		The proof has three steps: a tail bound for the highest real valuation, a decomposition of the auctioneer's profit into one marginal gain per shill, and the geometric choice of shill bids that makes every marginal gain from some index on positive.

		\paragraph*{Step 1: Tail bound for the order statistic.}
		Let $V := \max_{i \in N} v_i$ denote the highest real valuation.
		The valuations are \iid with $\pr{v_i \ge y} = 1/y$ for $y \le H$, so the CDF of $V$ satisfies $G(y) = (1 - 1/y)^n$ for $1 \le y \le H$.
		Since $(1 - x)^n \le 1 - nx + \binom{n}{2} x^2$ for $x \in [0, 1]$, the tail of $V$ satisfies, for $1 \le y \le H$,
		\begin{equation}\label{eq:tail}
			1 - G(y) = 1 - \left(1 - \frac{1}{y}\right)^n \ge \frac{n}{y} - \frac{n(n-1)}{2y^2}.
		\end{equation}

		\paragraph*{Step 2: Marginal profit analysis.}
		Consider shill bids $r_1 < \cdots < r_k$ with $1 < r_1$ and $r_k \le H$, all registered with the bidders' own distribution $F_H$, whose virtual value is $0$ on $[1, H]$ and which is regular, so that ties are settled on the raw bid; and put $r_0 := 1$ for the reserve price, which is not a shill and carries no penalty.
		The auctioneer aborts every shill that wins and reveals every shill that is outbid (\Cref{lem:reveal_outbid}), so the real winner pays at least the highest revealed floor, and each shill above $V$ wins in turn and forfeits $\pen$.
		The auctioneer's net revenue is therefore at least \[ \Pi = \max\{r_j : 0 \le j \le k,\ r_j \le V\} - \pen \cdot |\{j \ge 1 : r_j > V\}|, \] with equality when a single real bidder participates.
		We decompose the expected profit into marginal contributions.
		For each $j \ge 1$, define the $j$-th marginal gain as \[ \Delta_j := (r_j - r_{j-1}) \cdot \pr{V \ge r_j} - \pen \cdot \pr{V < r_j}.
		\]
		Since $\max\{r_j : r_j \le V\} = r_0 + \sum_{j=1}^{k} (r_j - r_{j-1}) \ind{V \ge r_j}$ and $|\{j \ge 1 : r_j > V\}| = \sum_{j=1}^{k} \ind{V < r_j}$, taking expectations gives $\e{\Pi} = 1 + \sum_{j=1}^k \Delta_j$.
		It therefore suffices to show that $\sum_{j=1}^k \Delta_j$ can be made arbitrarily large.

		\paragraph*{Step 3: Geometric construction.}
		Let $r_j = \alpha^j$ for some $\alpha > 1$ to be determined.
		Then $r_j - r_{j-1} = \alpha^j(1 - 1/\alpha)$, and applying the tail bound~\eqref{eq:tail} with $y = r_j$ and bounding $\pr{V < r_j}$ by $1$,
		\begin{align*}
			\Delta_j & \ge \alpha^j \left(1 - \frac{1}{\alpha}\right) \left(\frac{n}{\alpha^j} - \frac{n(n-1)}{2\alpha^{2j}}\right) - \pen \\
			         & \ge n\left(1 - \frac{1}{\alpha}\right) - \pen - \frac{n(n-1)}{2\alpha^{j}},
		\end{align*}
		where the second line uses $1 - 1/\alpha \le 1$.
		For the leading term to be strictly positive, we require \[ n\left(1 - \frac{1}{\alpha}\right) > \pen \quad \Longleftrightarrow \quad \alpha > \frac{n}{n - \pen}.
		\]
		This constraint admits a solution $\alpha > 1$ precisely when $\pen < n$.
		Fix such an $\alpha$ and put $\epsilon := n(1 - 1/\alpha) - \pen > 0$ and $J := \max\{1, \lceil \log_\alpha (n(n-1)/\epsilon) \rceil\}$; then $n(n-1)/(2\alpha^j) \le \epsilon/2$ and so $\Delta_j \ge \epsilon/2$ for all $j \ge J$, and $\Delta_j \ge -\pen$ for every $j$.
		Thus \[ \e{\Pi} = 1 + \sum_{j=1}^k \Delta_j \ge 1 - (J-1)\pen + (k - J + 1)\frac{\epsilon}{2}, \] which exceeds $n$ once $k \ge J + 2(n + (J-1)\pen)/\epsilon$.
		Fix such a $k$ and put $H_0 := \alpha^k$.
		For every $H \ge H_0$ all shills lie in $[1, H]$, so the tail bound~\eqref{eq:tail} applies at each $r_j$ and $\e{\Pi} > n \ge \OPT(\bm F)$, the honest revenue.
		This proves that the deviation is profitable.
	\end{proof}
}

\section{Conclusion}\label{sec:conclusion}

We introduced the Sequential Revelation Auction (SRA), a protocol that achieves credibility for any product distribution with vanishing revenue tails by combining a minimal MPC sub-computation---the Winner Determination Gadget---with sequential revelation and economic penalties.
The gadget reveals only the winner identity, never payment information, which prevents the adversary from exploiting the ``free option'' inherent in monolithic MPC approaches.
We proved that a penalty of $\pen \geq \sum_{i=1}^n \rev(F_i)$ is sufficient for credibility across all product distributions, and that it is tight: for equal-revenue distributions truncated at sufficiently large $H$, whose optimal revenue approaches the threshold, every smaller penalty admits a profitable deviation.

\paragraph*{Sharper penalties and richer adversaries.}
Our lower bound shows $\pen \ge \sum_{i=1}^n \rev(F_i)$ is necessary \emph{for the SRA}; whether another design---richer gadgets, or a different revelation order---is credible at a smaller penalty is open, as is tightening the constant in $\OPT(\bm F) \le \pen \le n \cdot \OPT(\bm F)$.
Our round bound holds on the equilibrium path only: the worst case grows linearly in the number of aborts, so relative to \citet{FerreiraEssaidi} the SRA gains a deterministic on-path bound and non-identical distributions without improving the worst case.
We also assume independent private values and a static adversary; correlation breaks both Myerson's benchmark and the leakage bound, and the repeated-auction setting, where the auctioneer's information advantage compounds across rounds, is outside the single-shot notion analyzed here.

\paragraph*{Practical implementation.}
Constant-round MPC with identifiable abort from MK-FHE or iO is currently impractical, so the SRA is at present a feasibility result rather than a deployable protocol.
The decomposition principle, however, is also a concrete efficiency argument: the gadget evaluates a single $\arg\max$ over $n$ committed values, a circuit of size $O(n\ell)$ for $\ell$-bit bids, rather than a full Myerson circuit with ironing and threshold-payment computation.
This suggests efficient instantiations may be within reach of modern MPC constructions such as garbled circuits or secret sharing over small fields, which is the most promising route to a practical credible auction.

\begin{standalonebib}
	\bibliographystyle{ACM-Reference-Format}
	\bibliography{mybib}
\end{standalonebib}

\bibliographystyle{ACM-Reference-Format}
\bibliography{mybib}


\begin{thebibliography}{25}


\ifx \showCODEN    \undefined \def \showCODEN     #1{\unskip}     \fi
\ifx \showISBNx    \undefined \def \showISBNx     #1{\unskip}     \fi
\ifx \showISBNxiii \undefined \def \showISBNxiii  #1{\unskip}     \fi
\ifx \showISSN     \undefined \def \showISSN      #1{\unskip}     \fi
\ifx \showLCCN     \undefined \def \showLCCN      #1{\unskip}     \fi
\ifx \shownote     \undefined \def \shownote      #1{#1}          \fi
\ifx \showarticletitle \undefined \def \showarticletitle #1{#1}   \fi
\ifx \showURL      \undefined \def \showURL       {\relax}        \fi
\providecommand\bibfield[2]{#2}
\providecommand\bibinfo[2]{#2}
\providecommand\natexlab[1]{#1}
\providecommand\showeprint[2][]{arXiv:#2}

\bibitem[Akbarpour and Li(2020)]%
        {AkbarpourLi}
\bibfield{author}{\bibinfo{person}{Mohammad Akbarpour} {and}
  \bibinfo{person}{Shengwu Li}.} \bibinfo{year}{2020}\natexlab{}.
\newblock \showarticletitle{Credible Auctions: A Trilemma}.
\newblock \bibinfo{journal}{\emph{Econometrica}} \bibinfo{volume}{88},
  \bibinfo{number}{2} (\bibinfo{year}{2020}), \bibinfo{pages}{425--467}.
\newblock
\href{https://doi.org/10.3982/ECTA15925}{doi:\nolinkurl{10.3982/ECTA15925}}


\bibitem[Baum et~al\mbox{.}(2020)]%
        {BaumOrsini20}
\bibfield{author}{\bibinfo{person}{Carsten Baum}, \bibinfo{person}{Emmanuela
  Orsini}, \bibinfo{person}{Peter Scholl}, {and} \bibinfo{person}{Eduardo
  Soria-Vazquez}.} \bibinfo{year}{2020}\natexlab{}.
\newblock \showarticletitle{Efficient Constant-Round MPC with Identifiable
  Abort and Public Verifiability}. In \bibinfo{booktitle}{\emph{Advances in
  Cryptology -- CRYPTO 2020}}. \bibinfo{publisher}{Springer},
  \bibinfo{address}{Cham}, \bibinfo{pages}{549--579}.
\newblock
\href{https://doi.org/10.1007/978-3-030-56880-1_20}{doi:\nolinkurl{10.1007/978-3-030-56880-1_20}}


\bibitem[Blum(1981)]%
        {Blum81}
\bibfield{author}{\bibinfo{person}{Manuel Blum}.}
  \bibinfo{year}{1981}\natexlab{}.
\newblock \showarticletitle{Coin flipping by telephone}. In
  \bibinfo{booktitle}{\emph{COMPCON Spring}}. \bibinfo{publisher}{IEEE},
  \bibinfo{address}{San Francisco, CA}, \bibinfo{pages}{133--137}.
\newblock


\bibitem[Bogetoft et~al\mbox{.}(2009)]%
        {Bogetoft09}
\bibfield{author}{\bibinfo{person}{Peter Bogetoft}, \bibinfo{person}{Dan~Lund
  Christensen}, \bibinfo{person}{Ivan Damg{\aa}rd}, \bibinfo{person}{Martin
  Geisler}, \bibinfo{person}{Thomas Jakobsen}, \bibinfo{person}{Mikkel
  Kr{\o}igaard}, \bibinfo{person}{Janus~Dam Nielsen},
  \bibinfo{person}{Jesper~Buus Nielsen}, \bibinfo{person}{Kurt Nielsen},
  \bibinfo{person}{Jakob Pagter}, {et~al\mbox{.}}}
  \bibinfo{year}{2009}\natexlab{}.
\newblock \showarticletitle{Secure multiparty computation goes live}. In
  \bibinfo{booktitle}{\emph{Financial Cryptography and Data Security (FC)}}.
  \bibinfo{publisher}{Springer}, \bibinfo{address}{Berlin, Heidelberg},
  \bibinfo{pages}{325--343}.
\newblock
\href{https://doi.org/10.1007/978-3-642-03549-4_20}{doi:\nolinkurl{10.1007/978-3-642-03549-4_20}}


\bibitem[Brandt(2006)]%
        {Brandt06}
\bibfield{author}{\bibinfo{person}{Felix Brandt}.}
  \bibinfo{year}{2006}\natexlab{}.
\newblock \showarticletitle{How to obtain full privacy in auctions}.
\newblock \bibinfo{journal}{\emph{International Journal of Information
  Security}} \bibinfo{volume}{5}, \bibinfo{number}{4} (\bibinfo{year}{2006}),
  \bibinfo{pages}{201--216}.
\newblock
\href{https://doi.org/10.1007/s10207-006-0001-y}{doi:\nolinkurl{10.1007/s10207-006-0001-y}}


\bibitem[Chitra et~al\mbox{.}(2024)]%
        {ChitraFerreira}
\bibfield{author}{\bibinfo{person}{Tarun Chitra}, \bibinfo{person}{Matheus
  V.~X. Ferreira}, {and} \bibinfo{person}{Kshitij Kulkarni}.}
  \bibinfo{year}{2024}\natexlab{}.
\newblock \showarticletitle{Credible, Optimal Auctions via Public Broadcast}.
  In \bibinfo{booktitle}{\emph{6th Conference on Advances in Financial
  Technologies (AFT 2024)}} \emph{(\bibinfo{series}{Leibniz International
  Proceedings in Informatics (LIPIcs)}, Vol.~\bibinfo{volume}{316})}.
  \bibinfo{publisher}{Schloss Dagstuhl -- Leibniz-Zentrum f{\"u}r Informatik},
  \bibinfo{address}{Dagstuhl, Germany}, \bibinfo{pages}{19:1--19:16}.
\newblock
\href{https://doi.org/10.4230/LIPIcs.AFT.2024.19}{doi:\nolinkurl{10.4230/LIPIcs.AFT.2024.19}}


\bibitem[Chung et~al\mbox{.}(2024)]%
        {ChungRoughgardenShi24}
\bibfield{author}{\bibinfo{person}{Hao Chung}, \bibinfo{person}{Tim
  Roughgarden}, {and} \bibinfo{person}{Elaine Shi}.}
  \bibinfo{year}{2024}\natexlab{}.
\newblock \showarticletitle{Collusion-Resilience in Transaction Fee Mechanism
  Design}. In \bibinfo{booktitle}{\emph{Proceedings of the 25th ACM Conference
  on Economics and Computation (EC)}}. \bibinfo{publisher}{ACM},
  \bibinfo{address}{New York, NY, USA}, \bibinfo{pages}{1045--1073}.
\newblock
\href{https://doi.org/10.1145/3670865.3673550}{doi:\nolinkurl{10.1145/3670865.3673550}}


\bibitem[Chung and Shi(2023)]%
        {ChungShi23}
\bibfield{author}{\bibinfo{person}{Hao Chung} {and} \bibinfo{person}{Elaine
  Shi}.} \bibinfo{year}{2023}\natexlab{}.
\newblock \showarticletitle{Foundations of Transaction Fee Mechanism Design}.
  In \bibinfo{booktitle}{\emph{Proceedings of the 2023 Annual ACM-SIAM
  Symposium on Discrete Algorithms (SODA)}}. \bibinfo{publisher}{SIAM},
  \bibinfo{pages}{3856--3899}.
\newblock
\href{https://doi.org/10.1137/1.9781611977554.ch150}{doi:\nolinkurl{10.1137/1.9781611977554.ch150}}


\bibitem[Ciampi et~al\mbox{.}(2022)]%
        {CiampiRaviSiniscalchiWaldner22}
\bibfield{author}{\bibinfo{person}{Michele Ciampi}, \bibinfo{person}{Divya
  Ravi}, \bibinfo{person}{Luisa Siniscalchi}, {and} \bibinfo{person}{Hendrik
  Waldner}.} \bibinfo{year}{2022}\natexlab{}.
\newblock \showarticletitle{Round-Optimal Multi-Party Computation with
  Identifiable Abort}. In \bibinfo{booktitle}{\emph{Advances in Cryptology --
  EUROCRYPT 2022}}. \bibinfo{publisher}{Springer}, \bibinfo{address}{Cham},
  \bibinfo{pages}{335--364}.
\newblock
\href{https://doi.org/10.1007/978-3-031-06944-4_12}{doi:\nolinkurl{10.1007/978-3-031-06944-4_12}}


\bibitem[Cleve(1986)]%
        {Cleve86}
\bibfield{author}{\bibinfo{person}{Richard Cleve}.}
  \bibinfo{year}{1986}\natexlab{}.
\newblock \showarticletitle{Limits on the security of coin flips when half the
  processors are faulty}. In \bibinfo{booktitle}{\emph{Proceedings of the
  Eighteenth Annual ACM Symposium on Theory of Computing (STOC)}}.
  \bibinfo{publisher}{ACM}, \bibinfo{address}{New York, NY, USA},
  \bibinfo{pages}{364--369}.
\newblock
\href{https://doi.org/10.1145/12130.12168}{doi:\nolinkurl{10.1145/12130.12168}}


\bibitem[Cohen et~al\mbox{.}(2024)]%
        {Cohen23}
\bibfield{author}{\bibinfo{person}{Ran Cohen}, \bibinfo{person}{Jack Doerner},
  \bibinfo{person}{Yashvanth Kondi}, {and} \bibinfo{person}{Abhi Shelat}.}
  \bibinfo{year}{2024}\natexlab{}.
\newblock \showarticletitle{Secure Multiparty Computation with Identifiable
  Abort via Vindicating Release}. In \bibinfo{booktitle}{\emph{Advances in
  Cryptology -- CRYPTO 2024}}. \bibinfo{publisher}{Springer},
  \bibinfo{address}{Cham}, \bibinfo{pages}{36--73}.
\newblock
\href{https://doi.org/10.1007/978-3-031-68397-8_2}{doi:\nolinkurl{10.1007/978-3-031-68397-8_2}}


\bibitem[Essaidi et~al\mbox{.}(2022)]%
        {FerreiraEssaidi}
\bibfield{author}{\bibinfo{person}{Meryem Essaidi}, \bibinfo{person}{Matheus
  V.~X. Ferreira}, {and} \bibinfo{person}{S.~Matthew Weinberg}.}
  \bibinfo{year}{2022}\natexlab{}.
\newblock \showarticletitle{Credible, Strategyproof, Optimal, and Bounded
  Expected-Round Single-Item Auctions for All Distributions}. In
  \bibinfo{booktitle}{\emph{Proceedings of the 13th Innovations in Theoretical
  Computer Science Conference (ITCS)}} \emph{(\bibinfo{series}{Leibniz
  International Proceedings in Informatics (LIPIcs)},
  Vol.~\bibinfo{volume}{215})}. \bibinfo{publisher}{Schloss Dagstuhl --
  Leibniz-Zentrum f{\"u}r Informatik}, \bibinfo{address}{Dagstuhl, Germany},
  \bibinfo{pages}{66:1--66:19}.
\newblock
\href{https://doi.org/10.4230/LIPIcs.ITCS.2022.66}{doi:\nolinkurl{10.4230/LIPIcs.ITCS.2022.66}}


\bibitem[Ferreira and Parkes(2023)]%
        {FerreiraParkes23}
\bibfield{author}{\bibinfo{person}{Matheus V.~X. Ferreira} {and}
  \bibinfo{person}{David~C. Parkes}.} \bibinfo{year}{2023}\natexlab{}.
\newblock \showarticletitle{Credible Decentralized Exchange Design via
  Verifiable Sequencing Rules}. In \bibinfo{booktitle}{\emph{Proceedings of the
  55th Annual ACM Symposium on Theory of Computing (STOC)}}.
  \bibinfo{publisher}{ACM}, \bibinfo{address}{New York, NY, USA},
  \bibinfo{pages}{723--736}.
\newblock
\href{https://doi.org/10.1145/3564246.3585233}{doi:\nolinkurl{10.1145/3564246.3585233}}


\bibitem[Ferreira and Weinberg(2020)]%
        {FerreiraWeinberg}
\bibfield{author}{\bibinfo{person}{Matheus V.~X. Ferreira} {and}
  \bibinfo{person}{S.~Matthew Weinberg}.} \bibinfo{year}{2020}\natexlab{}.
\newblock \showarticletitle{Credible, Truthful, and Two-Round (Optimal)
  Auctions via Cryptographic Commitments}. In
  \bibinfo{booktitle}{\emph{Proceedings of the 21st ACM Conference on Economics
  and Computation (EC)}}. \bibinfo{publisher}{ACM}, \bibinfo{address}{New York,
  NY, USA}, \bibinfo{pages}{683--702}.
\newblock
\href{https://doi.org/10.1145/3391403.3399495}{doi:\nolinkurl{10.1145/3391403.3399495}}


\bibitem[Ganesh and Zhang(2025)]%
        {GaneshZhang}
\bibfield{author}{\bibinfo{person}{Aadityan Ganesh} {and}
  \bibinfo{person}{Qianfan Zhang}.} \bibinfo{year}{2025}\natexlab{}.
\newblock \showarticletitle{Truthful, Credible, and Optimal Auctions for
  Matroids via Blockchains and Commitments}. In
  \bibinfo{booktitle}{\emph{Proceedings of the 26th ACM Conference on Economics
  and Computation (EC)}}. \bibinfo{publisher}{ACM}, \bibinfo{address}{New York,
  NY, USA}, \bibinfo{pages}{923--943}.
\newblock
\href{https://doi.org/10.1145/3736252.3742652}{doi:\nolinkurl{10.1145/3736252.3742652}}


\bibitem[Garg et~al\mbox{.}(2014)]%
        {GargGentryHaleviRaykova14}
\bibfield{author}{\bibinfo{person}{Sanjam Garg}, \bibinfo{person}{Craig
  Gentry}, \bibinfo{person}{Shai Halevi}, {and} \bibinfo{person}{Mariana
  Raykova}.} \bibinfo{year}{2014}\natexlab{}.
\newblock \showarticletitle{Two-Round Secure MPC from Indistinguishability
  Obfuscation}. In \bibinfo{booktitle}{\emph{Theory of Cryptography Conference
  (TCC)}}. \bibinfo{publisher}{Springer}, \bibinfo{address}{Berlin,
  Heidelberg}, \bibinfo{pages}{74--94}.
\newblock
\href{https://doi.org/10.1007/978-3-642-54242-8_4}{doi:\nolinkurl{10.1007/978-3-642-54242-8_4}}


\bibitem[Goldreich(2004)]%
        {Goldreich04}
\bibfield{author}{\bibinfo{person}{Oded Goldreich}.}
  \bibinfo{year}{2004}\natexlab{}.
\newblock \bibinfo{booktitle}{\emph{Foundations of Cryptography: Volume 2,
  Basic Applications}}. Vol.~\bibinfo{volume}{2}.
\newblock \bibinfo{publisher}{Cambridge University Press},
  \bibinfo{address}{Cambridge}.
\newblock
\href{https://doi.org/10.1017/CBO9780511721656}{doi:\nolinkurl{10.1017/CBO9780511721656}}


\bibitem[Ishai et~al\mbox{.}(2014)]%
        {IOZ14}
\bibfield{author}{\bibinfo{person}{Yuval Ishai}, \bibinfo{person}{Rafail
  Ostrovsky}, {and} \bibinfo{person}{Vassilis Zikas}.}
  \bibinfo{year}{2014}\natexlab{}.
\newblock \showarticletitle{Secure Multi-Party Computation with Identifiable
  Abort}. In \bibinfo{booktitle}{\emph{Advances in Cryptology -- CRYPTO 2014}}.
  Springer, \bibinfo{publisher}{Springer}, \bibinfo{address}{Berlin,
  Heidelberg}, \bibinfo{pages}{369--386}.
\newblock
\href{https://doi.org/10.1007/978-3-662-44381-1_21}{doi:\nolinkurl{10.1007/978-3-662-44381-1_21}}


\bibitem[Kuhn(1953)]%
        {Kuhn53}
\bibfield{author}{\bibinfo{person}{Harold~W. Kuhn}.}
  \bibinfo{year}{1953}\natexlab{}.
\newblock \showarticletitle{Extensive Games and the Problem of Information}.
\newblock In \bibinfo{booktitle}{\emph{Contributions to the Theory of Games,
  Volume II}}. \bibinfo{series}{Annals of Mathematics Studies},
  Vol.~\bibinfo{volume}{28}. \bibinfo{publisher}{Princeton University Press},
  \bibinfo{address}{Princeton, NJ}, \bibinfo{pages}{193--216}.
\newblock
\href{https://doi.org/10.1515/9781400881970-012}{doi:\nolinkurl{10.1515/9781400881970-012}}


\bibitem[Mukherjee and Wichs(2016)]%
        {MukherjeeWichs16}
\bibfield{author}{\bibinfo{person}{Pratyay Mukherjee} {and}
  \bibinfo{person}{Daniel Wichs}.} \bibinfo{year}{2016}\natexlab{}.
\newblock \showarticletitle{Two-Round Multiparty Computation via Multi-Key
  FHE}. In \bibinfo{booktitle}{\emph{Advances in Cryptology -- EUROCRYPT
  2016}}. \bibinfo{publisher}{Springer}, \bibinfo{address}{Berlin, Heidelberg},
  \bibinfo{pages}{735--763}.
\newblock
\href{https://doi.org/10.1007/978-3-662-49896-5_26}{doi:\nolinkurl{10.1007/978-3-662-49896-5_26}}


\bibitem[Myerson(1981)]%
        {Myerson81}
\bibfield{author}{\bibinfo{person}{Roger~B. Myerson}.}
  \bibinfo{year}{1981}\natexlab{}.
\newblock \showarticletitle{Optimal Auction Design}.
\newblock \bibinfo{journal}{\emph{Mathematics of Operations Research}}
  \bibinfo{volume}{6}, \bibinfo{number}{1} (\bibinfo{year}{1981}),
  \bibinfo{pages}{58--73}.
\newblock
\href{https://doi.org/10.1287/moor.6.1.58}{doi:\nolinkurl{10.1287/moor.6.1.58}}


\bibitem[Naor et~al\mbox{.}(1999)]%
        {NaorPinkasSumner99}
\bibfield{author}{\bibinfo{person}{Moni Naor}, \bibinfo{person}{Benny Pinkas},
  {and} \bibinfo{person}{Reuban Sumner}.} \bibinfo{year}{1999}\natexlab{}.
\newblock \showarticletitle{Privacy Preserving Auctions and Mechanism Design}.
  In \bibinfo{booktitle}{\emph{Proceedings of the 1st ACM Conference on
  Electronic Commerce (EC)}}. \bibinfo{publisher}{ACM},
  \bibinfo{address}{Denver, CO, USA}, \bibinfo{pages}{129--139}.
\newblock
\href{https://doi.org/10.1145/336992.337028}{doi:\nolinkurl{10.1145/336992.337028}}


\bibitem[Roughgarden(2024)]%
        {Roughgarden24}
\bibfield{author}{\bibinfo{person}{Tim Roughgarden}.}
  \bibinfo{year}{2024}\natexlab{}.
\newblock \showarticletitle{Transaction Fee Mechanism Design}.
\newblock \bibinfo{journal}{\emph{J. ACM}} \bibinfo{volume}{71},
  \bibinfo{number}{4} (\bibinfo{year}{2024}), \bibinfo{pages}{1--25}.
\newblock
\href{https://doi.org/10.1145/3674143}{doi:\nolinkurl{10.1145/3674143}}


\bibitem[Shi et~al\mbox{.}(2023)]%
        {ShiChungWu23}
\bibfield{author}{\bibinfo{person}{Elaine Shi}, \bibinfo{person}{Hao Chung},
  {and} \bibinfo{person}{Ke Wu}.} \bibinfo{year}{2023}\natexlab{}.
\newblock \showarticletitle{What Can Cryptography Do for Decentralized
  Mechanism Design?}. In \bibinfo{booktitle}{\emph{Proceedings of the 14th
  Innovations in Theoretical Computer Science Conference (ITCS)}}
  \emph{(\bibinfo{series}{Leibniz International Proceedings in Informatics
  (LIPIcs)}, Vol.~\bibinfo{volume}{251})}. \bibinfo{publisher}{Schloss Dagstuhl
  -- Leibniz-Zentrum f{\"u}r Informatik}, \bibinfo{address}{Dagstuhl, Germany},
  \bibinfo{pages}{97:1--97:22}.
\newblock
\href{https://doi.org/10.4230/LIPIcs.ITCS.2023.97}{doi:\nolinkurl{10.4230/LIPIcs.ITCS.2023.97}}


\bibitem[Yao(1982)]%
        {Yao82}
\bibfield{author}{\bibinfo{person}{Andrew~C Yao}.}
  \bibinfo{year}{1982}\natexlab{}.
\newblock \showarticletitle{Protocols for secure computations}. In
  \bibinfo{booktitle}{\emph{23rd Annual Symposium on Foundations of Computer
  Science (FOCS)}}. \bibinfo{publisher}{IEEE}, \bibinfo{address}{Chicago, IL},
  \bibinfo{pages}{160--164}.
\newblock
\href{https://doi.org/10.1109/SFCS.1982.38}{doi:\nolinkurl{10.1109/SFCS.1982.38}}


\end{thebibliography}

\end{document}